\documentclass[10pt,conference,letterpaper]{IEEEtran}
\newtheorem{example}{\textbf{Example}}{\itshape}{\rmfamily}
\newtheorem{theorem}{\textbf{Theorem}}{\itshape}{\rmfamily}
\newtheorem{definition}{\textbf{Definition}}{\itshape}{\rmfamily}
 \newtheorem{proof}{\textbf{Proof}}{\itshape}{\rmfamily}
  \newtheorem{lemma}{\textbf{Lemma}}{\itshape}{\rmfamily}

\usepackage[utf8]{inputenc}
\usepackage[T1]{fontenc}
\usepackage{amsmath}
\usepackage{amssymb}

\usepackage{xcolor}
\definecolor{mygreen}{RGB}{0,159,57}
\definecolor{myyellow}{RGB}{255,200,0}
\definecolor{myorange}{RGB}{255,140,0}
\definecolor{myblue}{RGB}{0,50,255}
\colorlet{myred}{red!90}

\usepackage{tikz}
\usetikzlibrary{calc}

\usepackage{arydshln}

\usepackage{enumitem}  
\setlist{nolistsep} 

\usepackage{placeins} 

\usepackage{comment}

\newcommand{\sos}[1]{$\mathit{#1^{sos}}$}

\newcommand{\whilelang}{\mathcal{W}}

\newcommand{\ifThenElseStatement}{\text{if-then-else}}
\newcommand{\whileStatement}{\text{while}}

\newcommand{\while}[1]{\lstinline[mathescape]!#1!}

\newcommand{\pv}[1]{\texttt{#1}}
\newcommand{\pvi}[1]{$\mathtt{_{#1}}$}

\usepackage{ stmaryrd }

\newcommand{\trans}[1]{[#1]}

\newcommand{\limp}{\rightarrow}

\newcommand{\Land}{\bigwedge}

\newcommand{\union}{\cup}

\newcommand*{\qed}{\null\nobreak\hfill\ensuremath{\square}}%

\newcommand{\eql}{{\,\simeq\,}}
\newcommand{\neql}{\not\simeq}
\newcommand{\zero}{{\tt 0}}
\newcommand{\suc}{{\tt suc}}
\newcommand{\pred}{{\tt pred}}

\newcommand{\Nat}{\mathbb{N}}
\newcommand{\Int}{\mathbb{I}}
\newcommand{\natsort}{\mathbb{N}}
\newcommand{\intsort}{\mathbb{I}}
\newcommand{\natsig}{S_{\Nat}}
\newcommand{\intsig}{S_{\Int}}

\usepackage{bussproofs}

\newcommand{\tracelogic}{\mathcal{L}}

\newcommand{\tpsort}{\mathbb{L}}

\newcommand{\tpsig}{S_{Tp}}
\newcommand{\lastsig}{S_n}
\newcommand{\pvsig}{S_V}

\usepackage{algorithm}
\usepackage[noend]{algpseudocode}

\usepackage{listings}
\usepackage{xcolor}

\definecolor{keywords} {rgb}{0.34,0.55,0.33}
\definecolor{comment}  {rgb}{0.27, 0.48, 0.34}

\lstdefinelanguage{while} {
    morekeywords = {if, else, while, const, Int, func, skip},
    morecomment=[l]{//},
    morestring=[b]",
}

\lstdefinestyle{while} {
    language            = while,
    basicstyle          = \small\ttfamily,
    commentstyle        = \color{comment},
    keywordstyle        = \bfseries\color{keywords},
    tabsize             = 2,
    captionpos          = b,
    keepspaces          = true,
    breaklines          = true,
    framesep            = 3.5mm,
    framexleftmargin    = 2.5mm,
    numberstyle         = \ttfamily\small
}

\newcommand{\rapid}{\textsc{Rapid}}
\newcommand{\vampire}{\textsc{Vampire}}

\newcommand{\smtlib}{\textsc{smt-lib}}
\newcommand{\z}{\textsc{Z3}}
\newcommand{\cvc}{\textsc{CVC4}}
\newcommand{\spacer}{\textsc{Spacer}}
\newcommand{\freqhorn}{\textsc{FreqHorn}}
\newcommand{\seahorn}{\textsc{SeaHorn}}
\usepackage{url}

\usepackage{rotating}
\usepackage{multirow}
\usepackage{graphicx}

\usepackage{scalerel}
\usetikzlibrary{svg.path}

\definecolor{orcidlogocol}{HTML}{A6CE39}
\tikzset{
	orcidlogo/.pic={
		\fill[orcidlogocol] svg{M256,128c0,70.7-57.3,128-128,128C57.3,256,0,198.7,0,128C0,57.3,57.3,0,128,0C198.7,0,256,57.3,256,128z};
		\fill[white] svg{M86.3,186.2H70.9V79.1h15.4v48.4V186.2z}
		svg{M108.9,79.1h41.6c39.6,0,57,28.3,57,53.6c0,27.5-21.5,53.6-56.8,53.6h-41.8V79.1z M124.3,172.4h24.5c34.9,0,42.9-26.5,42.9-39.7c0-21.5-13.7-39.7-43.7-39.7h-23.7V172.4z}
		svg{M88.7,56.8c0,5.5-4.5,10.1-10.1,10.1c-5.6,0-10.1-4.6-10.1-10.1c0-5.6,4.5-10.1,10.1-10.1C84.2,46.7,88.7,51.3,88.7,56.8z};
	}
}

\newcommand\orcidicon[1]{\href{https://orcid.org/#1}{\mbox{\scalerel*{
				\begin{tikzpicture}[yscale=-1,transform shape]
				\pic{orcidlogo};
				\end{tikzpicture}
			}{|}}}}

\usepackage{hyperref}

\title{
  Trace Logic for Inductive Loop Reasoning
}

\author{
	\IEEEauthorblockN{
		Pamina Georgiou\orcidicon{0000-0003-4856-4596},
		Bernhard Gleiss\orcidicon{0000-0002-2592-124X},  
		Laura Kov{\'a}cs\orcidicon{0000-0002-8299-2714}
	}
	\IEEEauthorblockA{TU Wien, Austria}
}

\usepackage{natbib}
\usepackage{graphicx}

\begin{document}

\maketitle

\begin{abstract}
  We propose trace logic, an instance of many-sorted first-order logic, to automate the partial correctness verification of programs containing loops. Trace logic generalizes semantics of program locations and captures loop semantics by encoding properties at arbitrary timepoints and loop iterations. We guide and automate inductive loop reasoning in trace logic by using generic trace lemmas capturing inductive loop invariants.
 Our work is implemented in the \rapid{} framework, by extending and integrating superposition-based first-order reasoning within \rapid. We successfully used \rapid{} to 
 prove correctness of many programs whose functional behavior are best summarized in the first-order theories of linear integer arithmetic, arrays and inductive data types. 
\end{abstract}

\textit{Related Version -- A compact, peer-reviewed version of this paper is published in the conference proceedings of Formal Methods in Computer-Aided Design (FMCAD) 2020.}

\section{Introduction}

One of the main challenges in automating software verification comes with handling inductive reasoning over programs containing loops. Until recently,
 automated reasoning in formal verification was the primary
 domain of satisfiability modulo theory (SMT) solvers~\cite{Z3,CVC4},
 yielding powerful advancements for inferring and proving loop properties with linear arithmetic and limited use of quantifiers, see e.g.~\cite{karbyshev2015property,gurfinkel2018quantifiers,fedyukovich2019quantified}. 
Formal verification however also requires
reasoning about unbounded data types, such
as arrays, and inductively defined data types. 
Specifying, for example as shown in Figure~\ref{fig:running},
that
every element in the array $\pv{b}$
is initialized by a non-negative array element of $\pv{a}$
requires reasoning with quantifiers and can be best expressed in many-sorted extensions of 
first-order logic. 
Yet, the recent progress in automation for
quantified reasoning in first-order theorem proving has not yet been fully integrated in formal verification.
In this paper we address such a use of first-order reasoning and propose trace logic $\tracelogic$, an instance of many-sorted first-order logic,
to automate the partial correctness verification of program loops, by expressing program semantics in $\tracelogic$, and use $\tracelogic$ in combination with superposition-based first-order theorem proving.

\paragraph*{\bf Contributions} 
In our previous work~\cite{barthe2019verifying}, an initial version of trace logic $\tracelogic$ was introduced to formalize and prove relational properties.
In this paper, we go beyond~\cite{barthe2019verifying} and turn trace logic $\tracelogic$ into an efficient approach to loop (safety) verification.
We propose trace logic $\tracelogic$ as a unifying framework to reason about both relational and safety properties expressed in full first-order logic with theories. We bring the following contributions. 

(i) We generalize the semantics of program locations by treating them as functions of execution timepoints. 
In essence, unlike other works
~\cite{bjorner2015horn,kobayashi2020fold,chakraborty2020verifying,ish2020putting}, we formalize program properties at arbitrary timepoints of locations. 

(ii) Thanks to this generalization, we provide a non-recursive axiomatization of program semantics in trace logic $\tracelogic$ and prove completeness of our axiomatization with respect to Hoare logic.
Our semantics in trace logic $\tracelogic$ supports arbitrary quantification over loop iterations (Section~\ref{sec:axiomatic-semantics}).

(iii) We guide and automate inductive loop reasoning in trace logic $\tracelogic$, by using generic trace lemmas capturing inductive loop invariants (Section~\ref{sec:verification}). We prove soundness of each trace lemma we introduce. 

(iv) We bring first-order theorem proving into the landscape of formal verification, by extending recent results in  superposition-based reasoning~\cite{gleiss2020subsumption,gleiss2020layered,kovacs2017coming} with support for trace logic properties, complementing SMT-based verification methods in the area (Section~\ref{sec:verification}).
As logical consequences of our trace lemmas are also loop invariants, superposition-based reasoning in trace logic $\tracelogic$ enables to automatically find loop invariants that are needed for proving safety assertions of program loops.

(v) We implemented our approach in the \rapid\ framework and combined \rapid\ with new extensions of the first-order theorem prover \vampire. 
We successfully evaluated  our work on more than 100 benchmarks taken from the SV-Comp repository~\cite{beyer2019automatic},  mainly consisting of safety verification challenges over programs containing arrays of arbitrary length and integers (Section~\ref{sec:implementation}). Our experiments show that \rapid{} automatically proves safety of many examples that, to the best of our knowledge, cannot be handled by other methods.

\section{Running Example}\label{sec:running}

\begin{figure} 
	\begin{lstlisting}[mathescape]
	func main() {
		const Int[] a;
		
		Int[] b;
		Int i = 0;
		Int j = 0;
		while (i < a.length) {
			if (a[i] $\geq$ 0) {
				b[j] = a[i]; 
				j = j + 1: 
			}
			i = i + 1;
		}
	}
	assert ($\forall$k$_\Int. \exists$l$_\Int.  ((0 \leq$ k $<$j $\wedge$ a.length $\geq 0$)  $\qquad\qquad\rightarrow$ b(k) = a(l)))
	\end{lstlisting}
		\caption{Program copying positive elements from array \pv{a} to \pv{b}.}\vspace*{-.5em}
		\label{fig:running}
\end{figure}

We illustrate and motivate our work with Figure \ref{fig:running}.
This program iterates over a constant integer array \pv{a} of arbitrary length and copies positive values into a new array \pv{b}. 
We are interested in proving the safety assertion given at line~15: given
that the length $\pv{a.length}$ of \pv{a} is not negative, every
element in \pv{b} is an element from \pv{a}. Expressing such a
property requires 
alternations of quantifiers in the first-order theories of
linear integer arithmetic and arrays, as formalized in line~15. 
We write $k_{\Int}$ and
$l_{\Int}$ to specify  that $k,l$ are of sort integer $\intsort$. 

While the safety assertion of line~15 holds, proving correctness of
Figure~\ref{fig:running} is challenging for most state-of-the-art
approaches, such as
e.g.~\cite{gurfinkel2015seahorn,karbyshev2015property,gurfinkel2018quantifiers,fedyukovich2019quantified}.
The reason is that proving safety of Figure~\ref{fig:running}  needs
inductive invariants with existential/alternating quantification and
involves inductive reasoning over arbitrarily bounded loop
iterations/timepoints. In this paper we address these challenges as follows.

(i) We extend the
semantics of program locations to des\-cribe locations parameterized by
timepoints, allowing us to express values of program variables at
arbitrary program locations within arbitrary loop iterations.
We write for example $i(l_{12}(it)))$ to denote the value of program variable $\pv{i}$ at
location $l_{12}$ in a loop iteration $it$, where the location $l_{12}$
corresponds to the program line~12. We reserve the constant $end$ for
specifying the last program location $l_{15}$, that is line 15,
corresponding to a terminating program execution of
Figure~\ref{fig:running}. We then write 
$b(end, k)$ to capture the value of array $\pv{b}$ at timepoint $end$ and
position $k$. For simplicity, as $\pv{a}$ is a constant array, we
simply write $a(k)$ instead $a(end, k)$.

(ii) Exploiting the semantics of program locations, we formalize 
the safety assertion of line~15 in trace logic $\tracelogic$ as
follows: 
\begin{equation}\label{ex:running:prop2}
  \begin{array}{l}
   \hskip-3.5em \forall k_\Int. \exists l_\Int. \big( (0 \leq k < j(end) \wedge a.length \geq 0) \\\rightarrow
    b(end, k)  \eql a(l)\big)
    \end{array}
\end{equation}

(iii) We express the 
semantics of Figure~\ref{fig:running} as a set $\mathcal{S}$ of first-order formulas
in trace logic $\tracelogic$, encoding values and dependencies
among program variables at arbitrary loop iterations.  To this end, we extend $\mathcal{S}$
with so-called trace lemmas, to automate inductive reasoning in
trace logic $\tracelogic$. One such trace lemma exploits the semantics
of updates to $\pv{j}$, allowing us to infer that {\it every}  value of
$j$ between $0$ to $j(end$), and thus each position at which the array
\pv{b} has been updated, is given by {\it some} loop
iteration. Moreover, updates to $\pv{j}$ happen at different loop
iterations and thus a position $\pv{j}$ at which $\pv{b}$ is updated is
visited uniquely throughout Figure~\ref{fig:running}.

(iv) We finally establish validity of~\eqref{ex:running:prop2}, by
deriving~\eqref{ex:running:prop2} to be a logical
consequence of $\mathcal{S}$.

\section{Preliminaries}
We assume familiarity with standard
first-order logic with equality and sorts. We write 
$\eql$  for equality and  $x_{S}$ to denote that a
logical variable $x$ has sort $S$. 
We denote by $\intsort$ the set of
integer numbers and by $\mathbb{B}$ the boolean sort. The term algebra of natural numbers is denoted by 
$\natsort$, with constructors $\zero$ and successor $\suc$.
We also consider the symbols $\pred$ and $\leq$ as part of the
signature of $\natsort$, interpreted respectively as the predecessor
function and less-than-equal relation.

Let $P$ be a first-order formula  with one free
variable $x$ of sort $\natsort$. We recall the standard (step-wise) 
induction schema for natural numbers as being 

\begin{equation}\label{eq:sw:ind}
\Big(P(\zero) \land \forall x'_\Nat. \big(P(x') \limp P(\suc(x'))\big)
\Big) \limp \forall x_\Nat. P(x)
\end{equation}

In our work, we use a variation of the induction
schema~\eqref{eq:sw:ind} to reason about intervals of loop
iterations. Namely, we use the following schema of {\it bounded induction}
\begin{equation}\label{eq:tr:ind}\nonumber
\begin{array}{l}
\hspace*{-.75em}	\bigg(
		P({bl}) \land \hfill \text{ {\scriptsize (base case)}} \\[.25em]
\hspace*{-.75em}		  ~~\forall x'_\Nat.\Big(\big(bl\leq x' < br \land P(x') \big) \limp P(\suc(x'))\Big) 
		  \bigg) ~
		   \hfill \text{ {\scriptsize (inductive case)}} \\[.25em]
\hspace*{-.75em}	\limp 
		\forall x_\Nat. \Big(bl\leq x \leq br \limp P(x)\Big),
\end{array}
\end{equation}
 where $bl, br\in\natsort$ are term algebra expressions of $\natsort$,
 called respectively as left and right
 bounds of bounded induction.

\section{Programming Model  $\whilelang$} \label{sec:model}
We consider 
programs written in an imperative while-like programming language $\whilelang$. 
This section recalls terminology
from~\cite{barthe2019verifying}, however adapted to our setting of
safety verification. Unlike~\cite{barthe2019verifying},  we do not consider 
 multiple program traces in $\whilelang$.
 In
Section~\ref{sec:axiomatic-semantics}, we then introduce a generalized
program semantics in trace logic $\tracelogic$, extended with reachability predicates.

\begin{figure}[t!]
	\centering
	\begin{align*}
		\text{program}    :=& \text{ function}\\
		\text{function}    :=& \text{ \while{func main()\{} \text{context} \while{\}} }\\
		\text{subprogram} :=& \text{ statement} \mid \text{context}\\
		\text{statement} :=& \text{ atomicStatement}\\
		\mid& \text{ \while{if(} \text{condition} \while{)\{} \text{context} \while{\} else \{} \text{context} \while{\}}}\\
		\mid& \text{ \while{while(} \text{condition} \while{)\{} \text{context} \while{\}}}\\
		\text{context} :=& \text{ statement; ... ; statement}\\[-2em]
	\end{align*}
	\caption{Grammar of $\whilelang$.\vspace*{-2em}}
	\label{fig:while-grammar}
\end{figure}

Figure~\ref{fig:while-grammar} shows the (partial) grammar of our
programming model $\whilelang$,
emphasizing the  use of contexts to capture lists of statements.
An input program in $\whilelang$ has a single  \pv{main}-function, 
with arbitrary nestings of if-then-else conditionals and
while-statements.
We consider \textit{mutable and constant variables}, where
variables are either integer-valued numeric variables or arrays of such numeric variables. We include standard 
\textit{side-effect free expressions over booleans and integers}. 

\subsection{Locations and Timepoints}\label{sec:locations}

A program in $\whilelang$ is considered as sets of locations, with
each location corresponding to positions/lines of program
statements in the program. 
Given a program statement $\pv{s}$, we denote by 
$l_s$ its (program) location. We reserve the location
$l_{\textit{end}}$  to denote 
the end of a program.  
For programs with loops, some program locations might be revisited
multiple times. 
We therefore model locations $l_s$ corresponding to a statement
$\pv{s}$ as functions of \textit{iterations} when the respective location is
visited. For simplicity, we write
$l_s$ also for the functional representation of the location $l_s$ of
$\pv{s}$. 
We thus consider locations as timepoints of a program and treat them $l_s$ as being functions $l_s$ over
iterations. The target sort of locations $l_s$ is $\tpsort$. 
For each enclosing loop of a statement $\pv{s}$, 
the function symbol $l_s$ takes arguments of sort $\natsort$,
corresponding to loop iterations. Further, when $\pv{s}$ is a loop
itself, we also introduce a function symbol $n_s$ with argument and
target sort $\natsort$; intuitively, $n_s$ corresponds to the last
loop iteration of $\pv{s}$. 
We denote the set of all function symbols $l_s$ as $\tpsig$, whereas
the set of all function symbols $n_s$ is written as  $\textit{S}_{n}$.

\noindent\begin{example}\label{ex:locations}
We refer to program statements $\pv{s}$
by their (first) line number in Figure~\ref{fig:running}. 
Thus, $l_5$ encodes the timepoint corresponding to the first
assignment of $\pv{i}$ in the program (line~5). We write $l_{7}(\zero)$ and
$l_{7}(n_{7})$ to denote the timepoints of the first and last loop iteration, respectively. The timepoints  
$l_{8}(\suc(\zero))$ and $l_{8}(it)$ correspond to the beginning of
the loop body in the second and the $it$-th
loop iterations, respectively. \qed
\end{example}
\subsection{Expressions over Timepoints}\label{sec:tp}
We next introduce commonly used expressions over 
timepoints. 
For each  while-statement $\pv{w}$ of $\whilelang$, we introduce a
function  $it^\pv{w}$  that returns a unique variable of sort
$\natsort$ for $\pv{w}$, denoting loop iterations of \pv{w}. 

Let $w_1,\dots,w_k$  be the enclosing
loops for statement $\pv{s}$ and consider an arbitrary term $it$ of
sort $\natsort$. We define $tp_{\pv{s}}$ to be the expressions denoting 
the timepoints of statements \pv{s} as 
\begin{align*}
	tp_\pv{s} &:= l_s(it^{w_1}, \dots, it^{w_k})&&\text{ if $\pv{s}$ is non-while statement}\\
	tp_\pv{s}(it) &:= l_s(it^{w_1}, \dots, it^{w_k}, it) &&\text{ if $\pv{s}$ is while-statement}\\
	\mathit{lastIt}_\pv{s} &:= n_s(it^{w_1}, \dots, it^{w_k})&&\text{ if $\pv{s}$ is while-statement}
\end{align*}
If $\pv{s}$ is a while-statement, we also introduce $lastIt_\pv{s}$
to denote the last iteration of \pv{s}. 
Further, consider an arbitrary subprogram $\pv{p}$, that is, \pv{p} is either a statement or a context.
The timepoint $\mathit{start}_\pv{p}$ (parameterized by an iteration
of each enclosing loop) denotes the timepoint when the execution of
$\pv{p}$ has started and is defined as
$$\mathit{start}_\pv{p} := 
\begin{cases}
	tp_\pv{p}(\zero) &\text{ if \pv{p} is while-statement}\\
	tp_\pv{p} &\text{ if \pv{p} is non-while statement}\\
	\mathit{start}_{\pv{s}_1} &\text{ if \pv{p} is context \pv{s}\pvi{1};$\dots$;\pv{s}\pvi{k}}
\end{cases}
$$

We also introduce the timepoint $\mathit{end}_\pv{p}$ to denote the
timepoint upon which a subprogram $\pv{p}$ has been completely
evaluated and define it as

$$
\mathit{end}_\pv{p}:=
\begin{cases}	
  \mathit{start}_{\pv{s}}  &\text{if \pv{s} occurs after $\pv{p}$ in a context}\\
  \mathit{end}_{\pv{c}}  &\text{if \pv{p} is last statement in context \pv{c}}\\
	\mathit{end}_{\pv{s}} &\text{if \pv{p} is context of if-branch or } \\& \text{else-branch of \pv{s}}\\
	\mathit{tp}_\pv{s}(\suc(it^{s})) &\text{if \pv{p} is context of body of \pv{s}} \\
	l_\mathit{end} &\text{if $\pv{p}$ is top-level context}\\
\end{cases}
$$

Finally, if $s$ is the topmost statement of the top-level context in \pv{main()}, we define $$\mathit{start} := \mathit{start}_s.$$

\subsection{Program Variables}\label{sec:prgVars}
We express values of program variables $\pv{v}$ at various timepoints of the program execution. To this end, we model
(numeric) variables $\pv{v}$ as functions 
$v: \tpsort \mapsto \intsort,$
 where $v(tp)$ gives the value of $\pv{v}$ at timepoint $tp$.
 For array variables $\pv{v}$,
 we add an additional argument of sort $\intsort$, corresponding to
 the position where the array is accessed; that is, 
$ v: \tpsort \times \intsort \mapsto \intsort$.
The set of such function symbols corresponding to program variables is
denoted by $\pvsig$. 

Our framework for constant, non-mutable variables can be simplified by
omitting the timepoint argument in the functional representation of 
such program variables, as illustrated below.

\begin{example}\label{ex:prgVars}
For Figure \ref{fig:running}, we denote by $i(l_5)$ the value of program variable 
$\pv{i}$ before being assigned in line 5. 

As the array variable $\pv{a}$ is non-mutable (specified by \pv{const} in the program), we write $a(i(l_{8}(it)))$ for 
the value of array $\pv{a}$ at the position corresponding to the current value of $\pv{i}$ at timepoint $l_{8}(it)$.
For the  mutable array $\pv{b}$, we consider timepoints where
$\pv{b}$ has been updated and write $b(l_{9}(it),
j(l_{9}(it)))$ for the array $\pv{b}$ at position $\pv{j}$ at the
timepoint $l_{9}(it)$ during the loop.
\qed\\
\end{example}

We emphasize that we consider (numeric) program variables \pv{v} to be of sort
$\intsort$, whereas loop iterations $it$ are of sort  $\natsort$. 

\subsection{Program Expressions}\label{sec:exp}
Arithmetic constants and program expressions are modeled using integer
functions and predicates.
Let $\pv{e}$ be an arbitrary program expression and write $\llbracket \pv{e} \rrbracket (tp)$ to denote 
the value of the evaluation of $\pv{e}$ at timepoint $tp$. 

Let $v\in \pvsig$, that is a function $v$ denoting a program variable
$\pv{v}$. Consider $\pv{e},\pv{e}_1,\pv{e}_2$ to be program
expressions and 
let $tp_1, tp_2$ denote two timepoints. We define 
\begin{equation*} \label{eq:eq:expressions}
\begin{array}{l}
  Eq(v,tp_1,tp_2):= 
\\
\qquad\quad\left\{
		\begin{aligned}
			&\forall \mathit{pos}_{\Int}. \;\;v(tp_1, \mathit{pos}) \eql v(tp_2, \mathit{pos}), \hspace{-0.5em}, \text{ if \pv{v} is an array}\\&
			v(tp_1) \eql v(tp_2), \text{otherwise}
		\end{aligned}
              \right.
\end{array}
\end{equation*}

to denote that the program variable $\pv{v}$ has the same values at
$tp_1$ and $tp_2$.  

We further introduce
\begin{equation*}\label{eq:eqall:expressions}
  \textit{EqAll}(tp_1,tp_2) := \bigwedge_{v \in \pvsig} Eq(v,tp_1,tp_2)
\end{equation*}

to define that all program variables have the same values at
timepoints $tp_1$ and $tp_2$. We also define
\begin{equation*}
  \begin{array}{l}
	\mathit{Update}(v,e,tp_1,tp_2) := \\
\qquad\quad	 v(tp_2) \eql \llbracket \pv{e} \rrbracket(tp_1) 
	\land \bigwedge_{v' \in \pvsig \setminus \{v\}} Eq(v',tp_1,tp_2), 
  \end{array}
  \end{equation*}
asserting that the numeric program variable \pv{v} has been updated while
all other program variables $\pv{v'}$ remain unchanged. This definition is
further 
extended  to array updates as

\begin{equation*}
  \begin{array}{l}
	 \mathit{Update}\mathit{Arr}(v,e_1,e_2,tp_1,tp_2) :=\\
\qquad\quad \forall \mathit{pos}_{\Int}. \ (\mathit{pos} \neql \llbracket e_1\rrbracket(tp_1) \rightarrow 
		v(tp_2, \mathit{pos}) \eql v(tp_1, \mathit{pos}))
    \\
\qquad\quad \land \ v(tp_2, \llbracket e_1 \rrbracket(tp_1)) \eql \llbracket e_2 \rrbracket(tp_1) 	\\
\qquad\quad  \bigwedge_{v' \in \pvsig \setminus \{v\}} Eq(v',tp_1,tp_2). 	\nonumber 
  \end{array}
  \end{equation*}

  \begin{example}
	In Figure \ref{fig:running}, we refer to the value of $\pv{i+1}$ at timepoint $l_{12}(it)$ as
	$i(l_{12}(it))+1$. 
	Let $\pvsig^\textit{1}$ be the set of function symbols
        representing the program variables of Figure \ref{fig:running}.
      
	For an update of \pv{j} in line~10 at some iteration $it$, we derive 
	\begin{align*}
	Update&(j, \pv{j+1}, l_9(it), l_{10}(it))  :=  j(l_{10}(it)) \eql (j(l_9(it))+1) 
	\\& \land  \quad \bigwedge_{v' \in \pvsig^\textit{1} \setminus \{j\}} Eq(v', l_9(it), l_{10}(it)). 
	\end{align*}\qed
   \end{example}

\section{Axiomatic Semantics in Trace Logic $\tracelogic$
}\label{sec:axiomatic-semantics}

Trace logic $\tracelogic$ has been
introduced in~\cite{barthe2019verifying}, yet for the setting of
relational verification. 
In this paper we generalize the formalization 
of~\cite{barthe2019verifying} in three ways. First,  (i) we define
program semantics in a non-recursive manner using the $Reach$
predicate  to characterize the set of reachable locations within a
given program context (Section~\ref{sec:reach}). 
Second, and most importantly, (ii) we prove completeness of trace logic $\tracelogic$ with
respect to Hoare Logic (Theorem~\ref{thm:completeness}),  which could have not been achieved in the
setting of~\cite{barthe2019verifying}.
Finally, (iii) we introduce
the use of logic
$\tracelogic$ for safety verification
(Section~\ref{sec:verification}). 

\subsection{Trace Logic $\tracelogic$}
Trace logic $\tracelogic$ is an instance of
many-sorted first-order logic with equality. We define the signature
$\Sigma(\tracelogic)$  of trace logic as 
$$ \Sigma(\tracelogic) : = \natsig \union \intsig \union \tpsig \union \pvsig \union \lastsig,$$
containing the signatures of the theory of natural numbers (term
algebra) $\natsort$ and integers $\intsort$,
as well the respective sets of timepoints, program variables and last
iteration symbols as defined in section \ref{sec:model}.

We next define the semantics of $\whilelang$ in trace logic $\tracelogic$.

\subsection{Reachability and its Axiomatization}\label{sec:reach} 
We introduce a predicate $\mathit{Reach}: \tpsort \mapsto \mathbb{B}$
to capture the set of timepoints reachable in an execution and use
$\mathit{Reach}$ to define the axiomatic
semantics of $\whilelang$ in trace logic $\tracelogic$.
We define reachability $\mathit{Reach}$ as a predicate over
timepoints,
in contrast to defining reachability as a predicate over program
configurations such as in~\cite{hoder2012generalized, bjorner2015horn, fedyukovich2019quantified, ish2020putting}.

We axiomatize $\mathit{Reach}$ using trace
logic formulas as follows. 

\begin{definition}[$\mathit{Reach}$-predicate]
	For any context $c$, any statement \pv{s}, let
        $\mathit{Cond}_{s}$ be the expression denoting a potential
        branching condition in \pv{s}.
        We define 
	\begin{equation*}
          \begin{array}{l}
	\mathit{Reach}(\mathit{start}_c) := 
	\begin{cases}
		true,  \\ \quad \text{if \pv{c} is top-level context}\\
		\mathit{Reach}(\mathit{start}_{s}) \land \mathit{Cond}_{s}(\mathit{start}_{s}), \\ \quad\text{if \pv{c} is context of if-branch of \pv{s}}\\
		\mathit{Reach}(\mathit{start}_{s}) \land \neg\mathit{Cond}_{s}(\mathit{start}_{s}), \\ \quad\text{if \pv{c} is  context of else-branch of \pv{s}}\\
		\mathit{Reach}(\mathit{start}_{s}) \land it^{s}< \mathit{lastIt}_{s}, \\ \quad\text{if \pv{c} is context of body of \pv{s}}.
              \end{cases}
          \end{array}
    \end{equation*}

	For any non-while statement $\pv{s}^\prime$ occurring in 
	context \pv{c}, let
	$$\mathit{Reach}(\mathit{start}_{s^\prime}) := \mathit{Reach}(\mathit{start_c}),$$
	
	and for any while-statement $\pv{s}^\prime$ occurring in context \pv{c}, let 
	$$\mathit{Reach}(tp_{s^\prime}(it^{s^\prime})) := \mathit{Reach}(\mathit{start}_{c}) \land it^{s^\prime} \leq \mathit{lastIt}_{s^\prime}.$$

	Finally let $\mathit{Reach(\mathit{end})}:= true$.\qed
      \end{definition}

Note that our reachability predicate $\mathit{Reach}$ allows specifying
properties about intermediate timepoints
(since those properties can only hold if the referred timepoints are
reached) and supports reasoning about which locations are reached.

\subsection{Axiomatic Semantics of $\whilelang$}
We axiomatize the semantics of each program statement in $\whilelang$,
and define the semantics of a program  in $\whilelang$ as the conjunction of all these axioms.

\paragraph{Main-function}
Let $\pv{p\pvi{0}}$ be an arbitrary, but fixed program in
$\whilelang$; we give our
definitions relative to $\pv{p\pvi{0}}$. 
The semantics of $\pv{p\pvi{0}}$, denoted by $\llbracket \pv{p\pvi{0}}
\rrbracket$,  consists of a conjunction of one implication per statement, where each implication has the reachability of the start-timepoint of the statement as premise and the semantics of the statement as conclusion:
$$\llbracket \pv{p\pvi{0}} \rrbracket := \Land_{\pv{s} \text{ statement of \pv{p\pvi{0}}}}  \forall \mathit{enclIts}. \big(Reach(\mathit{start}_s) \limp \llbracket \pv{s} \rrbracket\big)$$

where $enclIts$ is the set of iterations $\{it^{w_1},\dots,
it^{w_n}\}$ of  all enclosing loops $w_1, \dots, w_n$ of some statement \pv{s} in $\pv{p\pvi{0}}$, and
the semantics $\llbracket \pv{s} \rrbracket$ of program statements $\pv{s}$ is defined as follows.

\paragraph{Skip}
Let $\pv{s}$ be a statement \while{skip}. Then 
\begin{equation}\label{semantics_skip}
\llbracket \pv{s}\rrbracket := EqAll(\mathit{end}_\pv{s},\mathit{start}_\pv{s})
\end{equation}

\paragraph{Integer assignments}
Let $\pv{s}$ be an assignment \while{v = e}, 
 where $\pv{v}$ is an integer-valued  program variable and
$\pv{e}$ is an expression.
The evaluation of $\pv{s}$ is performed in one step such that, after the evaluation, the variable $\pv{v}$ has the same value as $\pv{e}$ before the evaluation. 
All other variables remain unchanged and thus 
\begin{equation}\label{semantics_int_assign}
\llbracket \pv{s}\rrbracket := \mathit{Update}(v,e,\mathit{end}_s,\mathit{start}_s)
\end{equation}

\paragraph{Array assignments}
Consider $\pv{s}$ of the form \while{a[e$_1$] = e$_2$},
with $\pv{a}$ being an array variable and
$\pv{e\pvi{1}}, \pv{e\pvi{2}}$ being expressions.
The assignment is evaluated in one step. 
After the evaluation of $\pv{s}$, the array $\pv{a}$ contains the
value of $\pv{e\pvi{2}}$ before the evaluation at position
$\mathit{pos}$ corresponding to the value of $\pv{e\pvi{1}}$ before
the evaluation. The values at all other positions of $\pv{a}$ and all
other program variables remain unchanged and hence 
\begin{equation}\label{semantics_arr_assign}
\llbracket \pv{s}\rrbracket := \mathit{UpdateArr}(v,e_1,e_2,\mathit{end}_s,\mathit{start}_s)
\end{equation}

\paragraph{Conditional \ifThenElseStatement{} Statements}
Let $\pv{s}$ be \while{if(Cond)\{c$_1$\} else \{c$_2$\}}.
The semantics of $\pv{s}$ states that entering the if-branch and/or entering the else-branch 
does not change the values of the variables and we have 
\begin{subequations}
\begin{align}
\label{semantics_ite_1}
\llbracket \pv{s} \rrbracket := 
&		    && \llbracket \pv{Cond} \rrbracket (\mathit{start}_s) \rightarrow \mathit{EqAll}(\mathit{start}_\pv{c\pvi{1}},\mathit{start}_s)\\
\label{semantics_ite_2}
&\land	&\neg&\llbracket \pv{Cond} \rrbracket (\mathit{start}_s) \rightarrow \mathit{EqAll}(\mathit{start}_\pv{c\pvi{2}},\mathit{start}_s)
\end{align}
\end{subequations}
 where the semantics
$\llbracket \pv{Cond} \rrbracket$ of the expression $\pv{Cond}$ is according
to Section~\ref{sec:exp}. 

\paragraph{While-Statements}

Let $\pv{s}$ be the \whileStatement-statement \while{while(Cond)\{c\}}.
We refer to $\pv{Cond}$ as the \emph{loop condition}.
The semantics of $\pv{s}$ is captured by conjunction of the  following three  properties: 
(\ref{semantics_while_1}) the iteration $\mathit{lastIt}_s$ is the first iteration where $\pv{Cond}$ does not hold,
(\ref{semantics_while_2}) entering the loop body does not change the values of the variables, 
(\ref{semantics_while_3}) the values of the variables at the end of evaluating $\pv{s}$ 
are the same as the variable values 
at the loop condition location in iteration $\mathit{lastIt}_s$. As
such, we have 
\begin{subequations}
\begin{align}
\llbracket \pv{s} \rrbracket :=      &&&\forall it^s_{\Nat}. \; (it^s<\mathit{lastIt}_\pv{s} \rightarrow \llbracket \pv{Cond} \rrbracket (tp_\pv{s}(it^s))) \nonumber\\
&\land &&\neg \llbracket \pv{Cond} \rrbracket (tp(\mathit{lastIt}_\pv{s})) \label{semantics_while_1}\\
&\land &&\forall it^s_{\Nat}. \; (it^s<\mathit{lastIt}_\pv{s} \rightarrow \mathit{EqAll}(\mathit{start}_\pv{c},tp_\pv{s}(it^s)) \label{semantics_while_2}\\
&\land && \mathit{EqAll}(\mathit{end}_\pv{s}, tp_s(\mathit{lastIt}_\pv{s})) \label{semantics_while_3}
\end{align}
\end{subequations}

\subsection{Soundness and Completeness.}

The axiomatic semantics of $\whilelang$ in trace logic is sound. That
is, given a program $\pv{p}$ in $\whilelang$ and a trace logic
property $F \in \tracelogic$,
we have that any interpretation in $\tracelogic$ is a model of
$F$ according to the small-step operational semantics of
$\whilelang$. We conclude the next theorem - and refer
to Appendix~\ref{sec:soundness} for details.

\begin{theorem}[$\whilelang$-Soundness\label{thm:sound}]
	Let $\pv{p}$ be a program. Then the axiomatic semantics
        $\llbracket \pv{p} \rrbracket$ is sound with respect to
        standard small-step operational semantics.
        \qed\\[-.5em]
\end{theorem}

Next, we show that the axiomatic semantics of $\whilelang$ in trace
logic $\tracelogic$  is complete with respect to Hoare logic \cite{hoare1969axiomatic}, as follows.

Intuitively, a Hoare Triple $\{F_1\}\pv{p}\{F_2\}$ corresponds to the trace logic formula
\begin{equation}\label{eq:translation-hoare}
\hspace*{-.75em}	\forall
\mathit{enclIts}. \big(\mathit{Reach}(\mathit{start}_{p}) \limp
([F_1](\mathit{start}_p)\rightarrow [F_2](\mathit{end}_p))\big)\hspace*{-.75em}
\end{equation}

where the expressions $[F_1](\mathit{start}_p)$ and
$[F_2](\mathit{end}_p)$ denote the result of adding to each program
variable in $F_1$ and $F_2$ the timepoints $\mathit{start}_p$
respectively  $\mathit{end}_p$ as first arguments.
We therefore define that the axiomatic semantics of $\whilelang$ is \emph{complete with respect to Hoare logic}, if for any Hoare triple $\{F_1\}\pv{p}\{F_2\}$ valid relative to the background theory $\mathcal{T}$, the corresponding trace logic formula \eqref{eq:translation-hoare}	is derivable from the axiomatic semantics of $\whilelang$ in the background theory $\mathcal{T}$.
With this definition at hand, we get the following result, proved
formally in Appendix~\ref{sec:completeness}.
\begin{theorem}[$\whilelang$-Completeness with respect to Hoare logic\label{thm:completeness}]
	The axiomatic semantics of $\whilelang$  in trace logic is
        complete with respect to Hoare logic.
        \qed
\end{theorem}

\section{Trace Logic for Safety Verification} \label{sec:verification}

We now introduce the use of trace logic $\tracelogic$ for verifying safety
properties of $\whilelang$ programs. We consider safety properties
$F$ 
expressed in first-order logic with theories, as illustrated in
line~15 of Figure~\ref{fig:running}.
Thanks to soundness and completeness of the axiomatic semantics of
$\whilelang$, a partially correct program $\pv{p}$ with regard to $F$ can be
proved to be correct using the axiomatic semantics of $\whilelang$ in
trace logic $\tracelogic$.  
That is, we  assume termination and establish partial program
correctness. Assuming the existence of an iteration violating the loop
condition can be help backward reasoning and, in particular, automatic splitting of loop iteration intervals.

However, proving correctness of a
program $\pv{p}$  annotated with a safety property $F$
faces the reasoning challenges of the underlying logic, in our case
of trace logic. Due to the presence of loops in $\whilelang$, a
challenging
aspect in using trace logic for safety verification is to handle
{inductive reasoning} as induction cannot be generally expressed in
first-order logic. To circumvent the challenge of inductive reasoning
and automate
verification using trace logic, we introduce 

a set of first-order lemmas, called \emph{trace lemmas}, and extend 
the semantics of $\whilelang$ programs in trace logic with these trace lemmas.
Trace lemmas describe generic inductive properties over arbitrary loop
iterations and any logical consequence of trace lemmas yields a valid
program loop property as well. We next summarize our approach  to program verification
using trace logic and then address the challenge of inductive
reasoning in trace logic $\tracelogic$. 

\subsection{Safety Verification in Trace Logic} \label{sec:safety}
Given a program $\pv{p}$ in $\whilelang$ and a safety property $F$,

\begin{itemize}
\item[(i)] we
express program semantics $\llbracket \pv{p} \rrbracket$ in trace
logic $\tracelogic$, as given in Section~\ref{sec:axiomatic-semantics}; 
\item[(ii)] we formalize
the safety property in trace logic $\tracelogic$, that is we express $F$ by using
program 
variables as functions of locations and timepoints (similarly as
in~\eqref{ex:running:prop2}).
For simplicity, let us denote the trace
logic formalization of $F$ also by $F$; 
\item[(iii)] we introduce instances $\mathcal{T_L^{\pv{p}}}$
of  a set 
$\mathcal{T_L}$ of trace lemmas, by instantiating trace
lemmas with program variables, locations and timepoints of
$\pv{p}$;
\item[(iv)] to verify $F$,  we then show that $F$ is a logical
consequence of $\llbracket \pv{p} \rrbracket \wedge
\mathcal{T_L^{\pv{p}}}$;
\item[(v)] however to conclude that $\pv{p}$ is partially correct
with regard to $F$, two more challenges need to be addressed. First, in addition to Theorem~\ref{thm:sound}, 
soundness of our trace lemmas $\mathcal{T_L}$ needs to be established, implying that our trace lemma
instances $\mathcal{T_L^{\pv{p}}}$ are also sound. Soundness of
$\mathcal{T_L^{\pv{p}}}$ implies then validity of $F$, whenever
$F$ is proven to be a logical consequence of sound
formulas $\llbracket \pv{p} \rrbracket \wedge
\mathcal{T_L^{\pv{p}}}$. However, to ensure that $F$ is provable in trace
logic, as a second challenge we need to ensure that our
trace lemmas $\mathcal{T_L}$, and thus their instances $\mathcal{T_L^{\pv{p}}}$, are strong
enough to prove $\llbracket \pv{p} \rrbracket \wedge \mathcal{T_L^{\pv{p}}}\implies F$. That is, proving that $F$ is a safety
assertion of $\pv{p}$ in our setting requires finding a suitable set $\mathcal{T_L}$ of
trace lemmas. 
\end{itemize}
In the remaining of this
section, we address (v) and show that our trace lemmas $\mathcal{T_L}$
are 
sound consequences of bounded  induction (Section~\ref{sec:TL}). Practical
evidence for using our trace lemmas are  further given in Section~\ref{sec:experiments}. 
\subsection{Trace Lemmas $\mathcal{T_L}$ for Verification}\label{sec:TL}
Trace logic
properties support arbitrary quantification over timepoints and
describe values of program variables at arbitrary loop  iterations and
timepoints. We therefore can relate timepoints with values of program
variables in trace logic $\tracelogic$, allowing us to describe the value
distributions of program variables as functions of timepoints
throughout program executions. As such, trace logic $\tracelogic$ supports 
 
\begin{enumerate}
	\item[(1)] reasoning about the {\it existence} of a specific
          loop iteration, allowing us to split the range of loop
          iterations at a particular timepoint, based  on the safety
          property we want to prove. For example, we can express and derive loop iterations
          corresponding to timepoints where one program variable takes
          a specific value for {\it the first time during loop
            execution}; 
	\item[(2)] universal quantification over the array content and
          range of 
          loop iterations bounded by two arbitrary left and right
          bounds, allowing us 
          to apply instances of the induction
          scheme~\eqref{eq:tr:ind} within a range of loop iterations
        bounded, for example, by $it$ and $lastIt_s$ for some while-statement \pv{s}. 
      \end{enumerate}
      Addressing these benefits of trace logic, we 
      
      express generic patterns of inductive program properties as
      \emph{trace lemmas}.

      Identifying a  suitable set $\mathcal{T_L}$ of trace lemmas to automate
      inductive reasoning in trace logic $\tracelogic$ is however challenging and
      domain-specific. We propose  three
      trace lemmas for inductive reasoning over arrays and integers,
      by considering
      \begin{description}
        \item[(A1)] one trace lemma
  
      describing how values of program variables change during an
      interval of loop iterations;
      
      \item[(B1-B2)] two trace lemmas to describe the behavior of loop
        counters. 
      \end{description}
     
      We prove soundness of our trace lemmas - below we
      include only one proof and refer 
      to Appendix~\ref{sec:trace-lemmas} for further details.

\subsubsection*{\textbf{(A1) Value Evolution Trace Lemma}} \label{lemma:valueEvolution}
Let $\pv{w}$ be a while-statement, let $\pv{v}$ be a mutable program
variable and let $\circ$ be a reflexive and transitive relation -
that is $\eql$ or $\leq$ in the setting of trace logic. 
The \emph{value evolution trace lemma of $\pv{w}$, $\pv{v}$, and $\circ$} is defined as
\begin{equation}\label{eq:TLA1}
  \begin{array}{l}
    \forall bl_{\natsort}, br_{\natsort}.\\
    \bigg( 
		\forall it_{\natsort}. 
			\Big(( bl\leq it < br \land v(tp_\pv{w}(bl)) \circ v(tp_\pv{w}(it))) 
			\\  \qquad \qquad\limp v(tp_\pv{w}(bl)) \circ v(tp_\pv{w}(\suc(it)))\Big) \tag{A1}
	\\\quad\limp 
		\big( bl \leq br \limp v(tp_\pv{w}(br)) \circ v(tp_\pv{w}(br)) \big) 
\bigg) 
  \end{array}
  \end{equation}
In our work, the value evolution trace
lemma is mainly instantiated with
the equality predicate $\eql$
to conclude that the value of a variable does not change during a
range of loop iterations, provided that the variable value does not
change at any of the considered loop iterations.

\begin{example}
  For Figure~\ref{fig:running}, the value evaluation trace lemma (A1) 
  yields the property
\begin{equation*}\label{eq:running:evolution}
  \begin{array}{l}
\forall j_\intsort.\ \forall bl_\natsort.\ \forall br_\natsort.\ \\
\bigg( 
 	\forall it_\natsort. \Big( (bl \leq it < br \ \wedge \  b(l_8(bl), j) = b(l_8(it), j) ) \\
			\qquad \qquad \rightarrow b(l_8(bl), j) = b(l_8(s(it)), j)
\Big) \\ \rightarrow 
\big( bl \leq br \rightarrow b(l_8(bl), j) = b(l_8(br), j) \big)
\bigg), 
  \end{array}
\end{equation*}
which allows to prove that the value of $\pv{b}$ at some position
$\pv{j}$
remains the same from the timepoint $it$ the value was first set until
the end of program execution. That is, we derive $b(l_9(end), j(l_9(it))) = a(i(l_8(it)))$.\qed
\end{example}

We next prove soundness of our trace lemma~\eqref{eq:TLA1}.

\noindent{\bf Proof {\it (Soundness Proof of Value Evolution Trace
    Lemma~\eqref{eq:TLA1})}} 
	Let $bl$ and $br$ be arbitrary but fixed and assume that the 
        premise of the outermost implication of~\eqref{eq:TLA1}
        holds. That is, 
        \begin{equation}\label{eq:TLA1:Prem}
          \begin{array}{l}
\forall it_{\natsort}.  \big(( bl\leq it < br \land v(tp_\pv{w}(bl)) \circ v(tp_\pv{w}(it))) \\
			\qquad\quad\limp v(tp_\pv{w}(bl)) \circ
            v(tp_\pv{w}(\suc(it)))\big)
            \end{array}
\end{equation}
        We use the induction axiom
        scheme~\eqref{eq:tr:ind}  and consider its instance
        with $P(it) := v(tp_\pv{w}(bl)) \circ v(tp_\pv{w}(it))$,
        yielding the following instance of~\eqref{eq:tr:ind}:
	\begin{subequations}	
	\begin{align}
	\hspace*{-1em} 	&\Big( v(tp_\pv{w}(bl)) \circ v(tp_\pv{w}(it)) \quad
          \wedge \label{form:A1-a}\\
	\hspace*{-1em} 	&\quad\forall it_{\natsort}. 
			\big(( bl\leq it < br \land v(tp_\pv{w}(bl)) \circ v(tp_\pv{w}(it))) \label{form:A1-b}\\
	\hspace*{-1em} 	&	\qquad\qquad\limp v(tp_\pv{w}(bl)) \circ v(tp_\pv{w}(\suc(it)))\big) \Big)
		\nonumber\\
               \hspace*{-1em} &	\limp\forall it_{\natsort}. \Big(bl\leq it \leq br \limp v(tp_\pv{w}(bl)) \circ v(tp_\pv{w}(it))\Big)\label{form:A1-c}
	\end{align}
	\end{subequations} 
       Note that  the base case property~\eqref{form:A1-a} holds
       since  $\circ$ is reflexive. Further, the inductive
       case~\eqref{form:A1-b} holds also  since it is
       implied by~\eqref{eq:TLA1:Prem}. We thus derive 
       property~\eqref{form:A1-c}, and 
       in particular 
	$bl\leq br \leq br \limp v(tp_\pv{w}(bl)) \circ
        v(tp_\pv{w}(br))$. Since $\leq$ is reflexive, we conclude $bl\leq br \limp v(tp_\pv{w}(bl)) \circ
        v(tp_\pv{w}(br))$, proving thus our trace
        lemma~\eqref{eq:TLA1}.
        \qed

\vskip.5em

\subsubsection*{\textbf{(B1) Intermediate
    Value Trace Lemma}} \label{lemma:intermediateValue}
Let $\pv{w}$ be a while-statement and let $\pv{v}$ be a mutable program
variable. We call $\pv{v}$ to be \emph{dense} if the following holds:
\begin{align*}
 \mathit{De}&\mathit{nse}_{w,v} :=  \forall it_{\natsort}.
\Big(
	it<\mathit{lastIt}_\pv{w}
\limp\\
	&\big(
			v(tp_\pv{w}(\suc(it)))=v(tp_\pv{w}(it))
		 \ \lor \ \\ &
			v(tp_\pv{w}(\suc(it)))=v(tp_\pv{w}(it)) + 1
	\big)
\Big)
\end{align*}

The \emph{intermediate value trace lemma of $\pv{w}$ and $\pv{v}$} is
defined as
\begin{equation}\label{eq:TLB1}
  \begin{array}{l}
	 \forall  x_{\intsort}. \Big(
		\big(
			\mathit{Dense}_{w,v} \land
			v(tp_\pv{w}(\zero)) \leq x < v(tp_\pv{w}(\mathit{lastIt}_\pv{w}))
		\big) \limp \\
		\quad\quad \exists it_{\natsort}.
		\big(
				it < \mathit{lastIt}_\pv{w}
			\land\
				v(tp_\pv{w}(it)) \eql x \ \land \tag{B1}\\
			\qquad\quad\quad \		
				v(tp_\pv{w}(\suc(it))) \eql v(tp_\pv{w}(it)) + 1
		\big)
	\Big)
  \end{array}\hspace*{-2em}
  \end{equation}
The intermediate value trace lemma~\eqref{eq:TLB1} allows us conclude
that if the variable $\pv{v}$ is dense, and if the value $x$ is between the
value of $\pv{v}$ at the beginning of the loop and the value of $\pv{v}$ at the
end of the loop, then there is an iteration in the loop, where $\pv{v}$ has
exactly the value $x$ and is incremented.
This trace lemma is mostly used to find specific iterations
corresponding to positions $x$ in an array.

\begin{example}
In Figure~\ref{fig:running}, using trace lemma~\eqref{eq:TLB1} we
synthesize the iteration $it$ such that $b(l_9(it), j(l_9(it))) =
a(i(l_8(it)))$.\qed
\end{example}

\vskip.5em

\subsubsection*{\textbf{(B2) Iteration Injectivity Trace Lemma}}
Let $\pv{w}$ be a while-statement and let $\pv{v}$ be a mutable program variable.
The \emph{iteration injectivity trace lemma of $\pv{w}$ and $\pv{v}$} is
\begin{align*}
	\forall it^1_\Nat, it^2_\Nat. \Big(
		&\big(
			\mathit{Dense}_{w,v} \land
			v(tp_\pv{w}(\suc(it^1))) = v(tp_\pv{w}(it^1)) + 1  \\& \land
			it^1 < it^2 \leq \mathit{lastIt}_\pv{w}
		\big)\tag{B2}\\
		&\limp
		v(tp_\pv{w}(it^1)) \neql v(tp_\pv{w}(it^2))
	\Big)
\end{align*}

The trace lemma (B2) states that a strongly-dense variable visits each array-position at most once.
As a consequence, if each array position is visited only once in a loop, we know that its value has not changed after the first visit, and in particular the value at the end of the loop is the value after the first visit. 
\begin{example}
  Trace lemma (B2) is necessary in Figure~\ref{fig:running} to apply
  the value evolution trace lemma (A1) for \pv{b}, as we need to
  make sure we will never reach the same position of \pv{j}
  twice. \qed
  \end{example}

Based on the soundness of our trace lemmas, we conclude the next
result.

\begin{theorem}[Trace Lemmas and Induction]
Let $\pv{p}$ be a program. Let $L$ be a trace lemma for some
while-statement $\pv{w}$ of $\pv{p}$ and some variable $\pv{v}$ of
$\pv{p}$.
Then $L$ is a consequence of the bounded induction
scheme~\eqref{eq:tr:ind} and of the axiomatic semantics of $\llbracket
\pv{p} \rrbracket$ in trace logic $\tracelogic$.
\qed
\end{theorem}

\section{Implementation and Experiments}\label{sec:implementation}

\newcommand{\confBase}{\rapid{$^-$}}
\newcommand{\confAdv}{\rapid{$^*$}}

\subsection{Implementation}

We implemented our approach in the \rapid\ tool, written in C++ and available at \url{https://github.com/gleiss/rapid}.

\rapid{} takes as input a program in the while-language $\whilelang$ together with a property expressed in trace logic $\tracelogic$ using the \smtlib{} syntax~\cite{barrett2017smtlib}. 
\rapid\ outputs 
(i) the program semantics as in Section~\ref{sec:axiomatic-semantics}, 
(ii) instantiations of trace lemmas for each mutable variable and for each loop of the program, as discussed in Section~\ref{sec:TL}, and 
(iii) the safety property, expressed 
in trace logic $\tracelogic$ and encoded in the \smtlib{} syntax.

For establishing safety, we pass the generated reasoning task to the first-order theorem prover \vampire{}~\cite{kovacs2013first} to prove the safety property from the program semantics and the instantiated trace lemmas\footnote{We also established the soundness of each trace lemma instance separately by running additional validity queries with \vampire{}.}, as discussed in Section~\ref{sec:safety}.
\vampire{} searches for a proof by refuting the negation of the property based on saturation of a set of clauses with respect to a set of inference rules such as resolution and superposition.

In our experiments, we use a custom version\footnote{\url{https://github.com/vprover/vampire/tree/gleiss-rapid}} of \vampire{} with a timeout of 60 seconds, in two different configurations. 
On the one hand, we use a configuration \textsc{\confBase}, where we tune \vampire{} to the trace logic domain using (i) existing options and (ii) domain-specific implementation to guide the high-level proof search.
On the other hand, we use a configuration \textsc{\confAdv}, which extends \textsc{\confBase} with recent techniques from~\cite{gleiss2020subsumption, gleiss2020layered} improving theory reasoning in equational theories. As such, \textsc{\confAdv} represents the result of a fundamental effort to improve \vampire{}'s reasoning for software verification. In particular, theory split queues~\cite{gleiss2020layered} present a partial solution to the prevalent challenge of combining quantification and \emph{light-weight} theory reasoning,  
drastically improving first-order reasoning in applications of software verification, as shown next.

\subsection{Experimental Results}\label{sec:experiments}

\begin{table*}[]
	\label{table:results}
	\caption{Experimental results}
	\begin{minipage}{.34\linewidth}
		\centering
		\begin{tabular}{lcc|}
			\texttt{\textbf{Benchmark} }                   & \texttt{\textbf{\confBase}} & \texttt{\textbf{\confAdv}} \\
			\hline
			atleast\_one\_iteration\_0                     & $\checkmark$                                               & $\checkmark$                                                                 \\
			atleast\_one\_iteration\_1                     & $\checkmark$                                               & $\checkmark$                                                               \\
			find\_sentinel                                 & $\checkmark$                                               & $\checkmark$                                                                \\
			find1\_0                                       & -                                                          & $\checkmark$                                                          \\
			find1\_1                                       & -                                                          & $\checkmark$                                                         \\
			find2\_0                                       & -                                                          & $\checkmark$                                                                                      \\
			find2\_1                                       & $\checkmark$                                               & $\checkmark$                                                                                      \\
			indexn\_is\_arraylength\_0                     & $\checkmark$                                               & $\checkmark$                                                                                      \\
			indexn\_is\_arraylength\_1                     & -                                                          & $\checkmark$                                                                                      \\
			set\_to\_one                                   & $\checkmark$                                               & $\checkmark$                                                                                      \\
			str\_cpy\_3                                    & $\checkmark$                                               & $\checkmark$                                                                                      \\
			\hdashline
			both\_or\_none                                 & -                                                          & $\checkmark$                                                                                      \\
			check\_equal\_set\_flag\_1                     & -                                                          & $\checkmark$                                                                                      \\
			collect\_indices\_eq\_val\_0             & -                                                          & $\checkmark$                                                                                      \\
			collect\_indices\_eq\_val\_1             & -                                                          & $\checkmark$                                                                                      \\
			copy                                           & -                                                          & $\checkmark$                                                                                      \\
			copy\_absolute\_0                              & -                                                          & $\checkmark$                                                                                      \\
			copy\_absolute\_1                              & -                                                          & $\checkmark$                                                                                      \\
			copy\_nonzero\_0                               & -                                                          & $\checkmark$                                                                                      \\
			copy\_partial                                  & -                                                          & $\checkmark$                                                                                      \\
			copy\_positive\_0                              & -                                                          & $\checkmark$                                                                                      \\
			copy\_two\_indices                             & -                                                          & $\checkmark$                                                                                      \\
			find\_max\_0                                   & -                                                          & $\checkmark$                                                                                      \\
			find\_max\_2                                   & -                                                          & $\checkmark$                                                                                      \\
			find\_max\_from\_second\_0                     & -                                                          & -                                                                                                 \\
			find\_max\_local\_2                            & -                                                          & -                                                                                                 \\
			find\_max\_up\_to\_0                           & -                                                          & -                                                                                                 \\
			find\_max\_up\_to\_2                           & -                                                          & -                                                                                                 \\
			find\_min\_0                                   & -                                                          & $\checkmark$                                                                                      \\
			find\_min\_2                                   & -                                                          & $\checkmark$                                                                                      \\
			find\_min\_local\_2                            & -                                                          & -                                                                                                 \\
			find\_min\_up\_to\_0                           & -                                                          & -                                                                                                 \\
			find\_min\_up\_to\_2                           & -                                                          & -                                                                                                 \\
			find1\_4                                       & -                                                          & $\checkmark$                                                                                      \\
			find2\_4                                       & $\checkmark$                                               & $\checkmark$                                                                                      \\
		\end{tabular}
	\end{minipage}%
	\begin{minipage}{.34\linewidth}
		\centering
		\begin{tabular}{lcc|}
			\texttt{\textbf{Benchmark} }                   & \texttt{\textbf{\confBase}} & \texttt{\textbf{\confAdv}} \\
			\hline
			in\_place\_max                                 & -                                                          & $\checkmark$                                                                                      \\
			inc\_by\_one\_0                          & -                                                          & $\checkmark$                                                                                      \\
			inc\_by\_one\_1                          & -                                                          & $\checkmark$                                                                                      \\
			inc\_by\_one\_harder\_0                  & -                                                          & $\checkmark$                                                                                      \\
			inc\_by\_one\_harder\_1                  & -                                                          & $\checkmark$                                                                                      \\
			init                                           & -                                                          & $\checkmark$                                                                                      \\
			init\_conditionally\_0                         & -                                                          & $\checkmark$                                                                                      \\
			init\_conditionally\_1                         & -                                                          & $\checkmark$                                                                                      \\
			init\_non\_constant\_0                         & -                                                          & $\checkmark$                                                                                      \\
			init\_non\_constant\_1                         & -                                                          & $\checkmark$                                                                                      \\
			init\_non\_constant\_2                         & -                                                          & $\checkmark$                                                                                      \\
			init\_non\_constant\_3                         & -                                                          & $\checkmark$                                                                                      \\
			init\_non\_constant\_easy\_0                   & -                                                          & $\checkmark$                                                                                      \\
			init\_non\_constant\_easy\_1                   & -                                                          & $\checkmark$                                                                                      \\
			init\_non\_constant\_easy\_2                   & -                                                          & $\checkmark$                                                                                      \\
			init\_non\_constant\_easy\_3                   & -                                                          & $\checkmark$                                                                                      \\
			init\_partial                                  & -                                                          & $\checkmark$                                                                                      \\
			init\_prev\_plus\_one\_0                   & -                                                          & $\checkmark$                                                                                      \\
			init\_prev\_plus\_one\_1                   & -                                                          & $\checkmark$                                                                                      \\
			init\_prev\_plus\_one\_alt\_0      & -                                                          & $\checkmark$                                                                                      \\
			init\_prev\_plus\_one\_alt\_1      & -                                                          & $\checkmark$                                                                                      \\
			max\_prop\_0                                   & -                                                          & $\checkmark$                                                                                      \\
			max\_prop\_1                                   & -                                                          & $\checkmark$                                                                                      \\
			merge\_interleave\_0                           & -                                                          & -                                                                                                 \\
			merge\_interleave\_1                           & -                                                          & -                                                                                                 \\
			min\_prop\_0                                   & -                                                          & $\checkmark$                                                                                      \\
			min\_prop\_1                                   & -                                                          & $\checkmark$                                                                                      \\
			partition\_0                                   & -                                                          & $\checkmark$                                                                                      \\
			partition\_1                                   & -                                                          & $\checkmark$                                                                                      \\
			push\_back                                     & -                                                          & $\checkmark$                                                                                      \\
			reverse                                        & -                                                          & $\checkmark$                                                                                      \\
			str\_cpy\_0                                    & -                                                          & $\checkmark$                                                                                      \\
			str\_cpy\_1                                    & -                                                          & $\checkmark$                                                                                      \\
			str\_cpy\_2                                    & $\checkmark$                                               & $\checkmark$                                                                                      \\	
			swap\_0                                        & -                                                          & $\checkmark$                                                                                      \\
		\end{tabular}
	\end{minipage} 
	\begin{minipage}{.33\linewidth}
		\centering
		\begin{tabular}{lcc}
			\texttt{\textbf{Benchmark} }                   & \texttt{\textbf{\confBase}} & \texttt{\textbf{\confAdv}} \\
			\hline		
			swap\_1                                        & -                                                          & $\checkmark$                                                                                      \\
			vector\_addition                               & -                                                          & $\checkmark$                                                                                      \\
			vector\_subtraction                            & -                                                          & $\checkmark$                                                                                      \\
			\hdashline
			check\_equal\_set\_flag\_0                     & $\checkmark$                                               & $\checkmark$                                                                                      \\
			find\_max\_1                                   & -                                                          & -                                                                                                 \\
			find\_max\_from\_second\_1                     & -                                                          & -                                                                                                 \\
			find1\_2                                       & $\checkmark$                                               & $\checkmark$                                                                                      \\
			find1\_3                                       & $\checkmark$                                               & $\checkmark$                                                                                      \\
			find2\_2                                       & $\checkmark$                                               & $\checkmark$                                                                                      \\
			find2\_3                                       & $\checkmark$                                               & $\checkmark$                                                                                      \\
			\hdashline
			collect\_indices\_eq\_val\_2             & -                                                          & $\checkmark$                                                                                      \\
			collect\_indices\_eq\_val\_3             & -                                                          & -                                                                                                 \\
			copy\_nonzero\_1                               & -                                                          & $\checkmark$                                                                                      \\
			copy\_positive\_1                              & -                                                          & $\checkmark$                                                                                      \\
			find\_max\_local\_0                            & -                                                          & -                                                                                                 \\
			find\_max\_local\_1                            & -                                                          & -                                                                                                 \\
			find\_max\_up\_to\_1                           & -                                                          & -                                                                                                 \\
			find\_min\_1                                   & -                                                          & -                                                                                                 \\
			find\_min\_local\_0                            & -                                                          & -                                                                                                 \\
			find\_min\_local\_1                            & -                                                          & -                                                                                                 \\
			find\_min\_up\_to\_1                           & -                                                          & -                                                                                                 \\
			merge\_interleave\_2                           & -                                                          & -                                                                                                 \\
			partition\_2                                   & -                                                          & $\checkmark$                                                                                      \\
			partition\_3                                   & -                                                          & $\checkmark$                                                                                      \\
			partition\_4                                   & -                                                          & -                                                                                                 \\
			partition\_5                                   & -                                                          & $\checkmark$                                                                                      \\
			partition\_6                                   & -                                                          & -                                                                                                 \\
			partition-harder\_0                            & -                                                          & $\checkmark$                                                                                      \\
			partition-harder\_1                            & -                                                          & $\checkmark$                                                                                      \\
			partition-harder\_2                            & -                                                          & -                                                                                                 \\
			partition-harder\_3                            & -                                                          & -                                                                                                 \\
			partition-harder\_4                            & -                                                          & -                                                                                                 \\
			str\_len                                       & $\checkmark$                                               & $\checkmark$                                                                                      \\\\
			\hline
			\textbf{Total solved}                          & \textbf{15}                                                & \textbf{78}           
		\end{tabular}
	\end{minipage}

\end{table*}

We considered challenging Java- and C-like verification benchmarks from the SV-Comp repository \cite{beyer2019automatic}, containing the combination of loops and arrays.
We omitted those examples for which the task is to find bugs in form of counterexample traces, as well as those examples that cannot be expressed in our programming model $\whilelang$, such as examples with explicit memory management.
In order to improve the set of benchmarks, we also included additional challenging programs and functional properties.
As a result, we obtained benchmarks ranging over 45 unique programs with a total of 103 tested properties. Our benchmarks are available in the \rapid\ repository\footnote{\url{https://github.com/gleiss/rapid/tree/master/examples/arrays}}. 

We manually transformed those benchmarks into our input format. SV-Comp benchmarks encode properties featuring universal quantification by extending the corresponding program with an additional loop containing a standard C-like assertion.
For instance, the property 
$$\forall i_{\intsort}.\ 0 \leq i < a.length \rightarrow P(a(i, end))$$
would be encoded by extending the program with a loop
\begin{align*}
	&\texttt{for(int i = 0; i < a.length; i++)} \\
	&\texttt{assert(P(a[i]))}
\end{align*}
While this encoding loses explicit structure and results in a harder reasoning task, it is necessary as other tools do not support explicit universal quantification in their input language.
In contrast, our approach can handle arbitrarily quantified properties over unbounded data structures. We, thus, directly formulate universally quantified properties, without using any program transformations.

The results of our experiments are presented in Table 1. 
We divided the results in four segments in the following order: the first eleven problems are quantifier-free, the largest part of 62 problems are universally quantified, seven problems are existentially quantified, while the last 23 problems contain quantifier alternations.  
First, we are interested in the overall number of problems we are able to prove correct. 
In the configuration \confAdv{}, which represents our main configuration, \vampire{} is able to prove 78 out of 103 encodings. In particular, we verify Figure~\ref{fig:running}, corresponding to benchmark {\tt copy\_positive\_1}, as well as other challenging properties that involve quantifier alternations, such as {\tt partition\_5}.

Second, we are interested in comparing the results for configurations \confBase{} and \confAdv{}, in order to understand the importance of recently developed techniques from \cite{gleiss2020subsumption} and \cite{gleiss2020layered} for reasoning in the trace logic domain. While \confBase{} is only able to prove 15 out of 103 properties, \confAdv{} is able to prove 78 properties, that is, \confAdv{} improves over \confBase{} by 63 examples. 
Moreover, only \confAdv{} is able to prove advanced properties involving quantifier alternations. We therefore see that \confAdv{} drastically outperforms \confBase{}, suggesting that the recently developed techniques are essential for efficient reasoning in trace logic.

Third, we are interested in what kinds of properties \rapid\ can prove. It comes with no surprise that all quantifier-free instances could be proved. Out of 62 universally quantified properties, \rapid\ could establish correctness of 53 such properties.
More interestingly, \rapid\ proves 14 out of 30 benchmarks containing either existentially quantified properties or such with quantifier alternations. The benchmarks that could not be solved by \rapid\ are primarily universally and alternatingly quantified properties that need additional  trace lemmas relating values of multiple program variables. 

\textbf{Comparing with other tools.}  
We compare our work against other approaches in~\ref{sec:related}. Here,  we omit a direct comparison of \rapid\ with other tools for the following reasons: \\
(1) Our benchmark suite includes 62 universally quantified and 11 non-quantified properties that could technically be supported by state-of-the-art tools such as \spacer/\seahorn\ and \freqhorn. Our benchmarks, however, also include 30 benchmarks with existential (7 examples) and alternating quantification (23 examples) that these tools cannot handle. As these examples depend on invariants that are alternatingly or at least existentially quantified, we believe these other tools cannot solve these benchmarks, while \confAdv{} could solve 14 examples in this domain. \\
(2) In our preliminary work~\cite{barthe2019verifying}, we already compared our reasoning within \rapid\ against \z\ and \cvc. These experiments showed that due to the fundamental difference in handling variables as functions over timepoints in our semantics, \rapid\ outperformed SMT-based reasoning approaches. \\
(3) Our program semantics is different than the one used in Horn clause verification techniques. 

Concerning previous approaches with first-order reasoners, the benchmarks of~\cite{gleiss2018loop} represent a subset of 55 examples from our current benchmark suite: only 21 examples from our benchmark suite could be proved by~\cite{gleiss2018loop}. For instance, our example in Figure \ref{fig:running} could not be proven in \cite{gleiss2018loop}. We believe that our work can be combined with approaches from~\cite{kovacs2009finding, gleiss2018loop} to  non-trivial invariants and loop bounds from saturation-based proof search. Our work can, thus, complement existing tools in proving complex quantified properties.

\section{Related Work}\label{sec:related}
Our work is closely related to recent efforts in using first-order theorem provers for proving software properties~\cite{kovacs2009finding, gleiss2018loop}. While~\cite{gleiss2018loop} captures programs semantics
in the first-order language of extended expressions over loop iterations, in our work we further generalize the semantics of program locations and consider program expressions over loop iterations and arbitrary timepoints. Further, we introduce and prove trace lemmas to automate inductive reasoning based on bounded induction over loop iterations. Our generalizations in trace logic proved to be necessary to automate the verification of properties with arbitrary quantification, which could not be effectively achieved in~\cite{gleiss2018loop}. Our work is not restricted to reasoning about single loops as in~\cite{gleiss2018loop}. 

Compared to~\cite{barthe2019verifying}, we provide a non-recursive generalization of the axiomatic semantics of programs in trace logic, prove completeness of our axiomatization in trace logic, ensure soundness of our trace lemmas and use trace logic for safety verification.

In comparison to verification approaches based on program transformations~\cite{kobayashi2020fold, chakraborty2020verifying, yang2019lemma}, we do not require user-provided functions to transform program states to smaller-sized states~\cite{ish2020putting}, nor are we restricted to universal properties generated by symbolic executions~\cite{chakraborty2020verifying}. Rather, we use only three trace lemmas that we prove sound and automate the verification of first-order properties, possibly with alternations of quantifiers.

The works~\cite{Dillig10,Cousot11} consider
expressive abstract domains and limit the generation of universal
invariants to these domains, while supporting potentially more generic program
grammars than our $\whilelang$ language. Our work however can verify universal and/or existential first-order properties with theories, which is not the case in~\cite{kobayashi2020fold,chakraborty2020verifying,Dillig10,Cousot11}.  
Verifying universal loop properties with arrays by implicitly finding invariants is addressed in~\cite{gurfinkel2018quantifiers, fedyukovich2019quantified, komuravelli2015compositional, fedyukovich2017sampling, fedyukovich2018accelerating, matsushita2020rusthorn}, and
by using constraint horn clause reasoning within property-driven reachability analysis in~\cite{hoder2012generalized, cimatti2012software}. 

Another line of research proposes abstraction and lazy interpolation~\cite{alberti2012lazy, afzal2020veriabs}, as well as recurrence solving with SMT-based reasoning~\cite{rajkhowa2018extending}. Synthesis-based approaches, such as \cite{fedyukovich2019quantified}, are shown to be successful when it comes to inferring universally quantified invariants and proving program correctness from these invariants. 
Synthesis-based term enumeration is used also in~\cite{yang2019lemma} in combination with user-provided invariant templates. 
Compared to these works, we do not consider programs only as a sequence of states, but model program values as functions of loop iterations and timepoints.
We synthesize bounds on loop iterations and infer first-order loop invariants as logical consequences of our trace lemmas and program semantics in trace logic.

\section{Conclusion}

We introduced trace logic to reason about safety loop properties over arrays. Trace logic supports explicit timepoint reasoning to allow arbitrary quantification over loop iterations. We use trace lemmas as consequences of bounded induction to automated inductive loop reasoning in trace logic.
We formalize the axiomatic semantics of programs in trace logic and prove it to be both sound and complete.
We report on our implementation in the \rapid{} framework, allowing us to use superposition-based reasoning in trace logic for verifying challenging  verification examples. 
Generalizing our work to termination analysis and extending our programming language, and its semantics in trace logic, with more complex constructs are interesting tasks for future work.

\vspace{1em}

\textit{Acknowledgements.}
This work was funded by the ERC Starting Grant 2014 SYMCAR 639270, the ERC Proof of Concept Grant 2018 SYMELS 842066, the Wallenberg Academy Fellowship 2014 TheProSE, and the Austrian FWF research project W1255-N23.

\bibliographystyle{IEEEtranN}
\bibliography{IEEEabrv, references}

\newpage
\appendices

\section{Small-step operational semantics} \label{sec:sos-semantics}

Here, we give a small-step operational semantics of $\whilelang$. 
Our presentation is semantically equivalent to standard small-step operational semantics, but differs syntactically in several points, in order to simplify later definitions and theorems:
(i) we annotate while-statements with counters to ensure the uniqueness of timepoints during the execution,
(ii) we reference nodes in the program-tree to keep track of the
current location during the execution instead of using strings to
denote the remaining program, 
(iii) we avoid additional constructs like states or configurations, 
(iv) we keep the timepoints in the execution separated from the values of the program variables at these timepoints, and
(v) we evaluate expressions on the fly.

We start by formalizing single steps of the execution of the program as transition rules, as defined in Figure~\ref{fig:operational-semantics}.
Intuitively, the rules describe 
(i) how we move the location-pointer around on the program-tree and 
(ii) how the state changes while moving the location-pointer around.
Each rule consists of 
(i) a premise $\mathit{Reach}(tp_1)$ for some timepoint $tp_1$, 
(ii) an additional premise $F$ (omitted if $F$ is $\top$), the so-called \emph{side-condition}, which is an arbitrary trace-logic formula referencing only the timepoint $tp_1$, 
(iii) the first conjunct of the conclusion of the form $\mathit{Reach}(tp_2)$ for some timepoint $tp_2$, and 
(iv) the second conjunct of the conclusion, which again is an arbitrary trace-logic formula $G$ referencing only the timepoints $tp_1$ and $tp_2$.

Next, we formalize the possible executions of the program as a set of first-order interpretations, so-called \emph{execution interpretations}.
In a nutshell, execution interpretations can be described as follows.
Each possible execution of the program induces an interpretation. 
For each such execution, the predicate symbol $\mathit{Reach}$ is interpreted as the set of timepoints which are reached during the execution. 
The function symbols denoting values of program variables are interpreted according to the transition rules at the timepoints which are reached during the execution, and are interpreted arbitrarily at all other timepoints.

We construct execution interpretation iteratively, as follows: 
We move around the program as defined by the transition rules. 
Whenever we reach a new timepoint, we choose a program state $J'$,
such that the side-conditions of the transition rule are fulfilled,
and extend the current interpretation $J$ with $J'$. We furthermore collect all timepoints that we already reached in $I$. We stop as soon as we reach $\mathit{end}$. We then construct an execution interpretation as follows: we interpret $Reach$ as $I$, extend $J$ to an interpretation of $\pvsig$ by choosing an arbitrary state at any timepoint which we did not reach, and choose an arbitrary interpretation of the theory symbols according to the background theory.

\begin{definition}[Program state]
	A \emph{program state at timepoint $tp$} is a partial interpretation, which exactly contains (i) for each non-array variable $\pv{v}$ an interpretation of $v(tp)$ and (ii) for each array variable $\pv{a}$ and for each element $pos$ of the domain $\intsig$ an interpretation of $a(tp,pos)$.
\end{definition}

\begin{definition}[Execution interpretation]
	\label{def:execution-interpretation} 
		Let $p_0$ be a fixed program.
		Let $I,J$ be any possible result returned by the algorithm in Algorithm~\ref{alg:execution-interpretation}.
		Let $M$ be any interpretation, such that 
		(i) $\mathit{Reach}(tp)$ is true iff $tp \in I$, 
		(ii) $M$ is an extension of $J$, and 
	       (iii) $M$ interprets the symbols of the background theory according to the theory. 
		Then $M$ is called an \emph{execution interpretation of $p_0$}.
\end{definition}

\begin{figure}[]
	\begin{prooftree}
		\AxiomC{}
		\LeftLabel{[\sos{init}]}
		\UnaryInfC{$\mathit{Reach}(\mathit{start})$}
	\end{prooftree}

	Let $s$ be a \while{skip}.
	\begin{prooftree}
		\AxiomC{$\mathit{Reach}(\mathit{start}_{s})$}
		\LeftLabel{[$\mathit{skip}^{sos}$]}
		\UnaryInfC{$\mathit{Reach}(\mathit{end}_{s}) \land \mathit{EqAll}(\mathit{start}_{s}, \mathit{end}_{s})$ }
	\end{prooftree}

	Let $s$ be an assignment \while{v = e}.
	\begin{prooftree}
		\AxiomC{$\mathit{Reach}(\mathit{start}_{s})$}
		\LeftLabel{[\sos{asg}]}
		\UnaryInfC{$\mathit{Reach}(\mathit{end}_{s}) \land \mathit{Update}(v,e,\mathit{start}_{s}, \mathit{end}_{s})$ }
	\end{prooftree}

	Let $s$ be an array-assignment \while{v[e$_1$] = e$_2$}.
	\begin{prooftree}
		\AxiomC{$\mathit{Reach}(\mathit{start}_{s})$}
		\LeftLabel{[\sos{asg_{arr}}]}
		\UnaryInfC{$\mathit{Reach}(\mathit{end}_{s}) \land \mathit{UpdateArr}(v,e_1,e_2,\mathit{start}_{s}, \mathit{end}_{s})$ }
	\end{prooftree}
	
	Let $\pv{s}$ be \while{if(Cond)\{c}$_1$\while{\}else\{c}$_2$\while{\}}.
	\begin{prooftree}
		\AxiomC{$\mathit{Reach}(\mathit{start}_{s})$}
		\AxiomC{$ \llbracket \mathit{Cond} \rrbracket (\mathit{start}_{s}) $}
		\LeftLabel{[\sos{ite_T}]}
		\BinaryInfC{$\mathit{Reach}(\mathit{start}_{c_1}) \land \mathit{EqAll}(\mathit{start}_{s}, \mathit{start}_{c_1})$}
	\end{prooftree}
	\begin{prooftree}
		\AxiomC{$\mathit{Reach}(\mathit{start}_{s})$}
		\AxiomC{$ \neg \llbracket \mathit{Cond} \rrbracket (\mathit{start}_{s}) $}
		\LeftLabel{[\sos{ite_F}]}
		\BinaryInfC{$\mathit{Reach}(\mathit{start}_{c_2}) \land \mathit{EqAll}(\mathit{start}_{s}, \mathit{start}_{c_2})$}
	\end{prooftree}

	Let $\pv{s}$ be \while{while(Cond)\{c\}}.
	\begin{prooftree}
		\AxiomC{$\mathit{Reach}(tp_{s}(it^{s}))$}
		\AxiomC{$\llbracket \mathit{Cond} \rrbracket (tp_{s}(it^{s})) $}
		\LeftLabel{[\sos{while_T}]}
		\BinaryInfC{$\mathit{Reach}(\mathit{start}_{c}) \land \mathit{EqAll}(tp_{s}(it^{s}), \mathit{start}_{c})$}
	\end{prooftree}
	\begin{prooftree}
		\AxiomC{$\mathit{Reach}(tp_{s}(it^{s}))$}
		\AxiomC{$ \neg \llbracket \mathit{Cond} \rrbracket (tp_{s}(it^{s})) $}
		\LeftLabel{[\sos{while_F}]}
		\BinaryInfC{$\mathit{Reach}(\mathit{end}_{s}) \land \mathit{EqAll}(tp_{s}(it^{s}), \mathit{end}_s)$}
	\end{prooftree}
	\caption{Small-step operational semantics using $\mathit{tp}$, $\mathit{start}$, $\mathit{end}$.}
	\label{fig:operational-semantics}
\end{figure}

\begin{algorithm}
	\caption{Algorithm to compute execution interpretation.}
	\label{alg:execution-interpretation}
	\begin{algorithmic}
		\State $\mathit{curr} = \mathit{start}$
		\State $I = \{ curr \}$
		\State $J = $ choose program state at $\mathit{curr}$
		\While{$\mathit{curr} \neq \mathit{end}$}
			\State choose $r :=$
			\AxiomC{$\sigma\mathit{Reach}(\mathit{curr})$}
			\AxiomC{$\sigma F$}
			\BinaryInfC{$\sigma\mathit{Reach}(\mathit{next}) \land \sigma G$}\DisplayProof, with $J \vDash \sigma F$
			\If{$r$ is [\sos{while_F}] for some statement $s$}
				\State $J = J \cup \{ \sigma\mathit{lastIt}_s \mapsto \sigma it^s\}$
			\EndIf
			\State choose a program state $J'$ such that $J \cup J' \vDash \sigma G$.
			\State $J = J \cup J'$
			\State $I = I \cup \{\mathit{next}\}$
			\State $\mathit{curr} = \mathit{next}$ 
		\EndWhile
		\State\Return $I, J$
	\end{algorithmic}
\end{algorithm}

With the definition of execution interpretations at hand, we are now able to define the valid properties of a program as the properties which hold in each execution interpretation.
\begin{definition}
Let $p_0$ be a fixed program. Let $F$ be a trace logic formula. Then $F$ is called \emph{valid} with respect to $p_0$, if $F$ holds in each execution interpretation of $p_0$.
\end{definition}

We conclude this subsection with stating simple properties of executions. 
The (i) first property states that whenever we reach the start of the execution of a subprogram \pv{p}, we also reach the end of the execution of \pv{p}. 
The (ii) second property states that whenever we reach the start of the execution of a context \pv{c}, we also reach the start of the execution of each statement occurring in \pv{c}. 
The (iii) third property states that whenever we reach the start of the execution of a \whileStatement-statement \pv{s}, then 
(a) we also reach the loop-condition check of \pv{s} in each iteration up to and including the last iteration, and 
(b) we also reach the start of the execution of the context of the loop body of \pv{s} in each iteration before the last iteration.
Formally, we have the following result.

\begin{lemma}\label{lemma:reach}
	Let $\pv{p}_0$ be a fixed program and let $M$ be an execution interpretation of $\pv{p}_0$. Let further $\pv{p}$ be an arbitrary subprogram of $\pv{p}_0$ and $\sigma$ be an arbitrary grounding of the enclosing iterations of $\pv{p}$ such that $\sigma\mathit{Reach}(\mathit{start}_p)$ holds. Then:
	\begin{enumerate}
		\item \label{lemma:reach1} $\sigma\mathit{Reach}(\mathit{end}_p)$ holds in $M$.
		\item \label{lemma:reach-context} 
		If \pv{p} is a context, $\sigma\mathit{Reach}(\mathit{start}_{s_i})$ holds in $M$ for any statement \pv{s}\pvi{i} occurring in \pv{p}.
		\item \label{lemma:reach2} If \pv{p} is a \whileStatement-statement \while{while(Cond)\{c\}}, then
		\begin{enumerate}[label={\alph*.}]
			\item $\sigma \mathit{Reach}(tp_p(it^p))$ holds in $M$ for any iteration $it^p \leq \sigma(\mathit{lastIt}_p)$.\label{lemma:reach2a}
			\item $\sigma\sigma' \mathit{Reach}(start_c)$ holds in $M$ for any iteration with $\sigma'it^p<\sigma(\mathit{lastIt}_p)$, where $\sigma'$ is any grounding of $it^p$.\label{lemma:reach2b}
		\end{enumerate}
	\end{enumerate}
\end{lemma}
\begin{proof}
We prove all three properties using a single induction proof.
  We proceed by structural induction over the program structure with the induction hypothesis $$\forall \mathit{enclIts}. \big(\mathit{Reach}(\mathit{start}_p) \limp \mathit{Reach}(\mathit{end}_p)\big).$$
	Let $\pv{p}$ be an arbitrary subprogram of $\pv{p}_0$. For an arbitrary grounding $\sigma$ of the enclosing iterations assume that $\sigma \mathit{Reach}(\mathit{start}_p)$ holds in $M$. 
	In order to show that $\sigma \mathit{Reach}(\mathit{end}_p)$ holds in $M$, 
	we perform a case distinction on the type of $\pv{p}$:
	\begin{itemize}
		\item Assume $\pv{p}$ is \while{skip}, or an integer- or array-assignment: Since $\sigma \mathit{Reach}(\mathit{start}_p)$ holds in $M$, the rule \sos{skip} resp. \sos{asg} resp. \sos{asg_{arr}} applies, so $\sigma \mathit{Reach}(\mathit{end}_p)$ holds in $M$ too.
		\item Assume $\pv{p}$ is a context $s_1;\dots;s_k$. By definition $\mathit{start_p}=\mathit{start_{s_1}}$, therefore $\sigma \mathit{Reach}(\mathit{start}_{s_1})$ holds in $M$. By the induction hypothesis, we know that
		$\sigma \mathit{Reach}(\mathit{start}_{s_i}) \limp \sigma \mathit{Reach}(\mathit{end}_{s_i})$ holds in $M$ for any $1\leq i \leq k$. Using a trivial induction, we conclude that $\sigma \mathit{Reach}(\mathit{end}_{s_i})$ holds in $M$ for any $1 \leq i \leq k$.
		\item Assume $\pv{p}$ is \while{if(Cond)\{c$_1$\} else \{c$_2$\}}. Assume w.l.o.g. that $\sigma\llbracket \mathit{Cond} \rrbracket(\mathit{start}_p)$ holds in $M$. Then the rule \sos{ite_T} applies, so $\sigma \mathit{Reach}(\mathit{start}_{c_1})$ holds in $M$. Using the induction hypothesis, we get $\sigma \mathit{Reach}(\mathit{start}_{c_1}) \limp \sigma \mathit{Reach}(\mathit{end}_{c_1})$, so $\sigma \mathit{Reach}(\mathit{end}_{c_1})$ holds in $M$. By definition, $\mathit{end}_c=\mathit{end}_p$, so $\sigma\mathit{Reach(\mathit{end}_p)}$ holds in $M$.
		\item Assume $\pv{p}$ is \while{while(Cond)\{c\}}. 
		We perform  bounded induction over $it^p$ from $\zero$ to $\sigma(\mathit{lastIt}_p)$ with the induction hypothesis $P(it^p) = \zero \leq \sigma(it^p) < \sigma(\mathit{lastIt}_p) \rightarrow \sigma\mathit{Reach}(tp_p(it^p))$. 
		
		The base case holds, since $\sigma\mathit{Reach}(\mathit{start}_p)$ is the same as $\sigma\mathit{Reach}(tp_p(\zero))$. 
		
		For the inductive case, assume that both $\sigma\sigma' \mathit{Reach}(tp_p(it^p))$ and $\sigma'(it^p) < \sigma(\mathit{lastIt}_p)$ holds for some grounding $\sigma'$ of $it^p$ with the goal of deriving $\sigma\sigma'\mathit{Reach}(tp_p(\suc(it^p)))$. Then rule \sos{while_T} applies, so $\sigma\sigma' \mathit{Reach}(\mathit{start}_c)$ holds in $M$. From the induction hypothesis, we conclude $\sigma\sigma' \mathit{Reach}(\mathit{start}_c)\limp \sigma\sigma' \mathit{Reach}(\mathit{end}_c)$, so $\sigma\sigma' \mathit{Reach}(\mathit{end}_c)$ holds. By definition, $\mathit{end}_c = \mathit{tp}_p(\suc(it^p))$, so we conclude that $\sigma\sigma' \mathit{Reach}(\mathit{tp}_p(\suc(it^p)))$ holds in $M$.
		
		We have established the base case and the inductive case, so we apply bounded induction to derive that 
		$$
		\forall it^p. 
		\big( 
			\sigma(it^p) \leq \sigma(\mathit{lastIt}_p) \limp \sigma\mathit{Reach}(tp_p(it^p))
		\big)
		$$
		holds in $M$.
		In particular,
                $\sigma\mathit{Reach}(\mathit{lastIt}_p)$ holds in
                $M$. Since by definition also $\sigma \neg \llbracket
                \mathit{Cond}\rrbracket(\mathit{lastIt}_p)$ holds in
                $M$, we deduce that \sos{while_F} applies, so
                $\sigma\mathit{Reach}(\mathit{end}_p)$ holds.
           \qed
	\end{itemize}
\end{proof}

\section{$\whilelang$-Soundness}
\label{sec:soundness}
We show that the axiomatic semantics introduced in Section~\ref{sec:axiomatic-semantics} is sound with respect to the operational semantics introduced in Appendix~\ref{sec:sos-semantics}.
Soundness is formalized as follows.

\begin{definition}[$\whilelang$-Soundness]
	Let $\pv{p}$ be a program and let $F$ be a trace logic formula.
	Then $F$ is called \emph{$\whilelang$-sound}, if for any execution interpretation $M$ we have $M \vDash F$.
\end{definition}

The following theorem states that the axioms defining the predicate $\mathit(Reach)$ are sound.
\begin{theorem}[$\whilelang$-Soundness of axioms defining $\mathit{Reach}$]
	For a given terminating program $\pv{p}_0$, the axioms defining $\mathit{Reach}$ are $\whilelang$-sound.
\end{theorem} 
\begin{proof}
	Let $M$ be an execution interpretation. 
	
	First, let \pv{c} be a context. We apply case distinction.
	\begin{itemize}
		\item Assume \pv{c} is the top-level context. From \sos{init} we conclude that $\mathit{Reach}(\mathit{start}_c)$ holds in $M$.
		\item Assume \pv{c} is the context of an if-branch of an \ifThenElseStatement-statement \pv{s} and assume that both $\sigma \mathit{Reach}(\mathit{start}_s)$ and $\sigma \mathit{Cond}_s(\mathit{start}_s)$ hold in $M$ for some grounding $\sigma$ of the enclosing iterations of $s$. Then rule \sos{ite_T} applies, from which we conclude that $\sigma\mathit{Reach}(\mathit{start}_c)$ holds in $M$.
		\item Assume \pv{c} is the context of an if-branch of an \ifThenElseStatement-statement \pv{s}. Analogously to the previous case.
		\item Assume \pv{c} is the context of the body of a
                  \whileStatement-statement \pv{s}, and assume that
                  $\sigma\mathit{Reach}(\mathit{start}_s)$ and
                  $\sigma'it^s < \sigma\mathit{lastIt}_{s}$ hold in
                  $M$ for some grounding $\sigma$ of the enclosing
                  iterations of $s$ and some grounding $\sigma'$ of
                  $it^s$. Using Lemma~\ref{lemma:reach} (case
                  \ref{lemma:reach2}.\ref{lemma:reach2b}) and the fact $\sigma'it^w < \sigma\mathit{lastIt}_{s'}$ we conclude that $\sigma\sigma'\mathit{Reach}(\mathit{start}_{c})$ holds in $M$.
	\end{itemize}
	
	Second, let \pv{s} be a non-\whileStatement-statement occurring
        in context \pv{c}. Assume further that
        $\sigma\mathit{Reach}(\mathit{start}_c)$ holds for some
        grounding $\sigma$ of the enclosing iterations of $c$. Using
        Lemma~\ref{lemma:reach} (case \ref{lemma:reach-context}), we conclude that $\sigma\mathit{Reach}(\mathit{start}_s)$ holds in $M$.

	Third, let \pv{s} be a \whileStatement-statement occurring in
        context \pv{c}. Assume further that both
        $\sigma\mathit{Reach}(\mathit{start}_c)$ and $\sigma' it^s
        \leq \sigma \mathit{lastIt}_s$ hold in $M$ for some grounding
        $\sigma$ of the enclosing iterations of $c$ and some grounding
        $\sigma$ of $it^s$. Using Lemma~\ref{lemma:reach} (case \ref{lemma:reach-context}) we conclude that $\sigma\sigma'\mathit{Reach}(\mathit{start}_s)$ holds in $M$.

	Finally, consider the last statement \pv{s} of the top-level
        context \pv{c}. From \sos{init} we conclude that
        $\mathit{Reach}(\mathit{start}_c)$ holds in $M$. From this, we
        conclude $\mathit{Reach}(\mathit{start}_s)$ using
        Lemma~\ref{lemma:reach} (case
        \ref{lemma:reach-context}). Finally, we apply
        Lemma~\ref{lemma:reach} (case \ref{lemma:reach1}) to conclude
        $\mathit{Reach}(\mathit{end}_s)$, which is the same as
        $\mathit{Reach}(\mathit{end}).$
        \qed
\end{proof}

We will now show that the axiomatic semantics of trace logic are $\whilelang{}$-sound.

\begin{theorem}[$\whilelang$-Soundness of Axiomatic Semantics of $\whilelang$]
	For a given terminating program $\pv{p\pvi{0}}$, the semantics $\llbracket \pv{p\pvi{0}} \rrbracket$ is $\whilelang$-sound.
\end{theorem} 

\begin{proof}
	Let $M$ be an execution interpretation of $\pv{p\pvi{0}}$. We have to show that for each statement \pv{s} of \pv{p\pvi{0}}, the formula 
	$$\forall \mathit{enclIts}. \big(\mathit{Reach}(\mathit{start}_s) \limp \llbracket s \rrbracket\big)$$
	holds in $M$.
	Let \pv{s} now be an arbitrary statement of $\pv{p\pvi{0}}$. For an arbitrary grounding $\sigma$ of the enclosing iterations assume that $\sigma \mathit{Reach}(\mathit{start}_s)$ holds in $M$. 
	In order to show that $\sigma\llbracket s \rrbracket$ holds in $M$, 
	we perform a case distinction on the type of the statement \pv{s}:
	\begin{itemize}
		\item Let \pv{s} be \while{skip}. Then $\sigma \mathit{Reach}(\mathit{start}_s)$ has been derived using \sos{skip}, so $\sigma \mathit{EqAll}(\mathit{start}_{s}, \mathit{end}_{s})$ holds in $M$, which is the same as $\sigma \llbracket s \rrbracket$.
	
		\item Let \pv{s} be \while{v = e}. Then $\sigma \mathit{Reach}(\mathit{start}_s)$ has been derived using \sos{asg}, so $\sigma \mathit{Update}(v,e,\mathit{start}_{s}, \mathit{end}_{s})$ holds in $M$, which is the same as $\sigma \llbracket s \rrbracket$.
		\item Let \pv{s} be \while{a[e$_1$] = e$_2$}. Then $\sigma \mathit{Reach}(\mathit{start}_s)$ has been derived using \sos{asg_{arr}}, so $\sigma \mathit{UpdateArr}(v,e_1,e_2,\mathit{start}_{s}, \mathit{end}_{s})$ holds in $M$, which is the same as $\sigma \llbracket s \rrbracket$.
	
		\item Let \pv{s} be \while{if(Cond)\{c$_1$\}else\{c$_2$\}}. Assume that $\sigma\llbracket \mathit{Cond} \rrbracket (\mathit{start}_{s})$ holds in $M$. Using \sos{ite_T}, we conclude that $\sigma\mathit{EqAll}(\mathit{start}_{c_1},tp_\pv{s})$ holds in $M$. In particular, formula~\eqref{semantics_ite_1} holds.
		
		Analogously we are able to prove that formula~\eqref{semantics_ite_2} holds in $M$. Combining both results, we conclude that $\sigma \llbracket s \rrbracket$ holds in $M$.
	
		\item Let $s$ be \while{while(Cond)\{p$_1$\}}.
		Formula~\eqref{semantics_while_1} defines $\sigma\mathit{lastIt}_s$ as the smallest iteration $it$ where $\sigma \llbracket \mathit{Cond} \rrbracket (tp_s(it)$ does not hold in $M$. Since we assume termination, such an iteration needs to exist, and in particular the definition is well-defined, so~\eqref{semantics_while_1} holds in $M$.
	
		Now let $it$ be an arbitrary iteration such that
                $it<\sigma\mathit{lastIt}_s$ holds in $M$. Using
                Lemma~\ref{lemma:reach} (case \ref{lemma:reach2}.\ref{lemma:reach2a}), we conclude that $\sigma \mathit{Reach}(tp_s(it))$ holds in $M$ from the fact that $\sigma \mathit{Reach}(\mathit{start}_s)$ holds in $M$.
		Since $\sigma\llbracket \mathit{Cond} \rrbracket (tp_s(it))$ holds in $M$ by the assumption $it<\sigma\mathit{lastIt}_s$, we know that $\sigma \mathit{Reach}(\mathit{start}_s)$ has been derived using \sos{while_T}, and in particular that $\sigma\mathit{EqAll}(\mathit{start}_c,tp_s(it)$ holds in $M$, which is the same as axiom~\eqref{semantics_while_2}.
		
		Finally, we obtain that $\sigma
                \mathit{Reach}(tp_s(\mathit{lastIt}_s))$ holds in $M$
                from the fact that $\sigma
                \mathit{Reach}(\mathit{start}_s)$ holds in $M$ using
                Lemma~\ref{lemma:reach}
                (case\ref{lemma:reach2}.\ref{lemma:reach2a}). 
		By definition of $\mathit{lastIt}_s$, the formula $\sigma \llbracket \mathit{Cond} \rrbracket (tp_s(\mathit{lastIt}_s))$ does not hold in $M$, so $\sigma \mathit{Reach}(\mathit{start}_s)$ has been derived using \sos{while_F}. In particular, $\sigma\mathit{EqAll}(\mathit{end}_\pv{s}, tp_s(\mathit{lastIt}_\pv{s}))$ holds in $M$, which is the same as~\eqref{semantics_while_3}.
		\qed
	\end{itemize}
\end{proof}

\section{Completeness}
\label{sec:completeness}
We how that trace logic semantics is complete with respect to Hoare logic.
We start by translating Hoare triples to trace logic formulas. 
Recall that a Hoare triple $\{F_1\}\pv{p}\{F_2\}$ denotes that if $F_1$ holds at the beginning of the execution of $\pv{p}$, then $F_2$ holds at the end of the execution of $\pv{p}$.
We write such a fact in trace logic as 
$\trans{F_1}(\mathit{start}_p)\rightarrow \trans{F_2}(\mathit{end}_p)$, 
where the expressions $\trans{F_1}(\mathit{start}_p)$ and $\trans{F_2}(\mathit{end}_p)$ denote the result of adding to each program variable in $F_1$ resp. $F_2$ the timepoint $\mathit{start}_p$ resp. $\mathit{end}_p$ as first argument. 
For example, consider the program $\pv{p\pvi{0}} := \pv{i=i+1}$. We can derive the Hoare triple $\{i\eql 2\}\pv{p\pvi{0}}\{i \eql 3\}$. In a similar way, we are able to derive the trace logic formula
$$
	i(\mathit{start}_p)\eql 2 
\limp  
	i(\mathit{end}_p)\eql 3.
$$
Additionally we have to deal with the technical complication that
Hoare logic overspecifies unreachable subprograms. Consider a program
$\pv{p\pvi{0}}$, containing $\pv{p} := \pv{i=i+1}$ as an
\emph{unreachable subprogram}. As Hoare logic does not take the
context of a subprogram into account, we can again derive a Hoare
triple $\{i\eql 2\}\pv{p}\{i \eql 3\}$, even though $\pv{p}$ is never
executed. In contrast, in trace logic we will only derive the more
precise formula 
$$\forall \mathit{enclIt}.
	\Big(
		\mathit{Reach}(\mathit{start}_p)
	\limp 
		\big(
			i(\mathit{start}_p)\eql 2 
		\limp  
			i(\mathit{end}_p)\eql 3 
		\big)
	\Big),$$
which takes the reachability of the subprogram $\pv{p}$ into account. Note that this difference only occurs for (strict) subprograms, as the start of a program is by definition always reachable.

\begin{definition}
	Let $\pv{p\pvi{0}}$ be a fixed program.
	\begin{itemize}
		\item Let $\trans{}$ be a function which translates any Hoare logic formula $F$ to a trace logic formula $F'$, where $F'$ is obtained by adding to each symbol $v$ denoting a program variable in $F$ as first argument the free variable $tp_{\tpsort}$. For any background theory $\mathcal{T}$, let further $\trans{\mathcal{T}} := \{\forall tp_{\tpsort}. \trans{F} \mid F \in \mathcal{T}\}$ be the translation of $\mathcal{T}$.
		\item Trace logic is called \emph{}{complete with respect to Hoare logic}, if for any fixed background theory $\mathcal{T}$ and for any Hoare triple $\{F_1\}p\{F_2\}$ provable using $\mathcal{T}$, the trace logic formula 
		$$\forall \mathit{enclIts}. \Big(\mathit{Reach}(\mathit{start}_p) \rightarrow \big(\trans{F_1}(\mathit{start}_p)\rightarrow \trans{F_2}(\mathit{end}_p)\big)\Big)$$
		is \emph{provable} from the trace logic axioms using the background theory $\trans{\mathcal{T}}$.
	\end{itemize}	
\end{definition}

We are now able to establish the completeness of trace logic with respect to Hoare logic.
\begin{theorem}
	Let $\pv{p\pvi{0}}$ be a fixed program. Then the trace logic semantics is complete with respect to Hoare logic.
\end{theorem}

\begin{proof}
	Let $\pv{p\pvi{0}}$ be a fixed terminating program. We proceed by structural induction on the Hoare calculus derivation with the induction hypothesis that for any subprogram $\pv{p}$ of $\pv{p\pvi{0}}$ and for any formulas $F_1,F_2$, if 
	$\{F_1\}p\{F_2\}$ is derivable in the Hoare calculus, then 
	$$
	\forall \mathit{enclIts}. 
	\Big(
		\mathit{Reach}(\mathit{start}_p) 
	\rightarrow 
		\big(
			\trans{F_1}(\mathit{start}_p)
		\rightarrow 
			\trans{F_2}(\mathit{end}_p)
		\big)
	\Big)
	$$ 
	is entailed by the trace logic semantics. 

	Consider now an arbitrary subprogram $\pv{p}$ of
        $\pv{p\pvi{0}}$ such that $\{F_1\}p\{F_2\}$ is derivable in
        Hoare logic.
        For an arbitrary grounding $\sigma$ of the enclosing iterations assume that $\sigma\mathit{Reach}(\mathit{start}_{p})$ holds. This fact together with the definition of the trace logic semantics implies that $\sigma\llbracket p \rrbracket$ holds. We now use a case distinction to show the implication 
	\begin{equation}
		\sigma \trans{F_1}(\mathit{start}_{p}) \limp \sigma \trans{F_2}(\mathit{end}_{p}).\label{proof:completeness-formula1}
	\end{equation}
	Since the grounding $\sigma$ is arbitrary, this then concludes the proof.
	\begin{itemize}
		\item Skip: Assume the last rule is 
		\begin{prooftree}
			\AxiomC{}
			\UnaryInfC{$\{F_1\} skip \{F_1\}$}
		\end{prooftree}		
		We have to show $\sigma \trans{F_1}(\mathit{start}_{p}) \limp \sigma \trans{F_1}(\mathit{end}_{p})$.
		The semantics $\sigma\llbracket p \rrbracket$ state that $\sigma EqAll(\mathit{start}_p, \mathit{end}_p)$ holds.
		Using this formula, we can rewrite $\sigma \trans{F_1}(\mathit{start}_{p})$ into $\sigma \trans{F_1}(\mathit{end}_{p})$, which shows that~\eqref{proof:completeness-formula1} holds.
		\item Assignment: Assume that the last rule is 
		\begin{prooftree}
			\AxiomC{}
			\UnaryInfC{$\{F_2[x \mapsto e]\} x:=e \{F_2\}$}
		\end{prooftree}	
		We have to show the implication $\sigma \trans{F_2}[x~\mapsto~e](\mathit{start}_p) \limp \sigma \trans{F_2}(\mathit{end}_p)$.
		By definition, $\sigma \llbracket p\rrbracket$ consists of $\sigma(x(\mathit{end}_p))=\sigma(\llbracket e \rrbracket(\mathit{start}_p))$ and of $\sigma(v(\mathit{end}_p))=\sigma(v(\mathit{start}_p))$ for all other variables $v$. Using these equations, we rewrite $\sigma \trans{F_2}[x~\mapsto~e](\mathit{start}_p)$ into $\sigma \trans{F_2}(\mathit{end}_p)$, which proves~\eqref{proof:completeness-formula1}.
		\item Weakening: Assume the last rule is
		\begin{prooftree}
			\AxiomC{$F_1 \limp F_1'$}
			\AxiomC{$\{F_1'\} p \{F_2'\}$}
			\AxiomC{$F_2' \limp F_2$}
			\TrinaryInfC{$\{F_1\} p \{F_2\}$}
		\end{prooftree}
		First, the formulas $F_1 \limp F_1'$ and $F_2' \limp F_2$ are tautologies in Hoare Logic. Since we assume that $\trans$ maps Hoare Logic tautologies to Trace Logic tautologies, we get that $\trans{F_1 \limp F_1'}(tp)$ and $\trans{F_2' \limp F_2}(tp)$ hold for arbitrary ground timepoints $tp$. In particular, $\trans{F_1 \limp F_1'}(\sigma(\mathit{start}_p))$ and $\trans{F_2' \limp F_2}(\sigma(\mathit{end}_p))$ hold, which can be written as $\sigma \trans{F_1}(\mathit{start}_p) \limp \sigma \trans{F_1'}(\mathit{start}_p)$ and $\sigma \trans{F_2'}(\mathit{end}_p) \limp \sigma \trans{F_2}(\mathit{end}_p)$. Second, we use the induction hypothesis and the assumption $\sigma \mathit{Reach}(\mathit{start}_p)$ to conclude that the trace logic axioms imply $\sigma \trans{F_1'}(\mathit{start}_p) \limp \sigma \trans{F_2'}(\mathit{end}_p)$. 
		Combining the three implications shows that~\eqref{proof:completeness-formula1} holds.
		\item Concatenation: Assume the last rule is
		\begin{prooftree}
			\AxiomC{$\{G_1\} p_1 \{G_1'\}$}
			\AxiomC{$\dots$}
			\AxiomC{$\{G_k\} p_k \{G_k'\}$}
			\TrinaryInfC{$\{G_1\} p_1;\dots;p_k \{G_k'\}$}
		\end{prooftree}
		where $G_1=F_1$ and $G_k'=F_2$.
		Using Lemma~\ref{lemma:reach} (case \ref{lemma:reach-context}), we conclude from $\sigma\mathit{Reach}(\mathit{start}_{p})$ that $\sigma\mathit{Reach}(\mathit{start}_{p_i})$ holds for any $1 \leq i \leq k$. Combining these facts with applications of the induction hypothesis yields that $\sigma \trans{G_i}(\mathit{start}_{p_i}) \limp \sigma \trans{G_i'}(\mathit{end}_{p_i})$ holds for any $1 \leq i \leq k$. Since $\sigma(\mathit{end}_{p_i}) = \sigma(\mathit{start}_{p_{i+1}})$ for any $1 \leq i < k$, we use a trivial induction to conclude 
		$$\sigma \trans{G_1}(\mathit{start}_{p}) \limp \sigma \trans{G_k'}(\mathit{end}_{p}).$$
		In particular, since $G_1=F_1$ and $G_k'=F_2$, we conclude that~\eqref{proof:completeness-formula1} holds.
		\item If-then-else conditionals: Assume that the last rule is 
 		\begin{prooftree}
			\AxiomC{$\{\llbracket \mathit{Cond} \rrbracket \land F_1\} p_1 \{F_2\}$}
			\AxiomC{$\{\neg \llbracket \mathit{Cond} \rrbracket \land F_1\} p_2 \{F_2\}$}
			\BinaryInfC{$\{F_1\}$ \while{if(Cond)\{p$_1$\}else\{p$_2$\}}$\{F_2\}$}
		\end{prooftree}
		W.l.o.g. assume that $\sigma \llbracket \mathit{Cond} \rrbracket(\mathit{start}_{p})$ holds. We assume that $\sigma \trans{F_1}(\mathit{start}_{p})$ holds with the goal of deriving $\sigma \trans{F_2}(\mathit{end}_{p})$.
		First, we combine $\sigma\llbracket p \rrbracket$ with $\sigma \llbracket \mathit{Cond} \rrbracket (\mathit{start}_{p})$ to derive $\sigma EqAll(\mathit{start}_{p},\mathit{start}_{p_1})$. From this we derive $\sigma \llbracket \mathit{Cond} \rrbracket(\mathit{start}_{p_1})$ and $\sigma \trans{F_1}(\mathit{start}_{p_1})$. 
		Second $\sigma \mathit{Reach}(\mathit{start}_{p})$ and $\sigma \llbracket \mathit{Cond} \rrbracket(\mathit{start}_{p})$ imply $\sigma \mathit{Reach}(\mathit{start}_{p_1})$. 
		We then combine the induction hypothesis 
		\begin{equation*}
			\begin{array}{l}
				\sigma \mathit{Reach}(\mathit{start}_{p_1}) \limp \\
					\qquad \Big( \sigma \big(\llbracket \mathit{Cond} \rrbracket \land \trans{F_1} \big)(\mathit{start}_{p_1}) \limp \sigma \trans{F_2}(\mathit{end}_{p_1})\Big)
			\end{array}
		\end{equation*}
		with $\sigma \mathit{Reach}(\mathit{start}_{p_1})$, $\sigma \llbracket \mathit{Cond} \rrbracket (\mathit{start}_{p_1})$ and $\sigma \trans{F_1} (\mathit{start}_{p_1})$ to obtain $\sigma \trans{F_2}(\mathit{end}_{p_1})$.
		Since $\mathit{end}_{p_1}=\mathit{end}_{p}$, we
                conclude $\sigma \trans{F_2}(\mathit{end}_{p})$, which
                proves~\eqref{proof:completeness-formula1}.
                
		\item While-statement: Assume that the last rule is
		\begin{prooftree}
			\AxiomC{$\{\mathit{Cond} \land F\} p_1 \{F\}$}
			\UnaryInfC{$\{F\} $\while{while(Cond)\{p$_1$\}}$\{\neg \mathit{Cond} \land F\}$}
		\end{prooftree}
		We again assume that $\sigma \trans{F_1}(\mathit{start}_{p})$ holds with the goal of deriving $\sigma \trans{F_1}(\mathit{end}_{p})$.

		We perform a bounded induction on $it^p$ from $0$ to $\sigma\mathit{lastIt}_p$ with the induction hypothesis $\sigma \trans{F_1}(tp_p(it^p))$.

		Base Case: The formula $\sigma \trans{F_1}(\mathit{start}_{p})$ holds and can be written as $\sigma \trans{F_1}(tp_{p}(\zero))$.
		
		Inductive Case: We have to show the implication 
		\begin{equation*}
			\begin{array}{l}
				\sigma \forall it^p. 
				\Big( 
					\big( 
						it^p < \mathit{lastIt}_p 
					\land 
						\trans{F_1}(tp_p(it^p))
					\big) \\ \qquad
				\limp 
					\trans{F_1}(tp_p(\suc(it^p)))
				\Big).
			\end{array}
		\end{equation*}
		
		Let $\sigma'$ be an extension of $\sigma$ with an arbitrary grounding of $it^p$, and assume that $\sigma'(it^p < \mathit{lastIt}_p )$ and $\sigma' \trans{F_1}(tp_p(it^p))$ hold. We now have to show $\sigma' \trans{F_1}(tp_p(\suc(it^p)))$. 
		Combining $\sigma'\llbracket p \rrbracket$ and $\sigma'(it^p < \mathit{lastIt}_p )$ yields both $\sigma'\llbracket \mathit{Cond} \rrbracket (tp_p(it^p))$ and $\sigma'\mathit{EqAll}(\mathit{start}_{p_1},tp_p(it^p)$. We use the latter fact first to rewrite the former fact to $\sigma'\llbracket \mathit{Cond} \rrbracket (\mathit{start}_{p_1})$ and second to rewrite $\sigma' \trans{F_1}(tp_p(it^p))$ to  $\sigma' \trans{F_1}(\mathit{start}_{p_1})$.
		Third, we obtain $\sigma'
                \mathit{Reach}(\mathit{start}_{p_1})$ using
                Lemma~\ref{lemma:reach} (case \ref{lemma:reach2}.\ref{lemma:reach2b}).
		
		The induction-hypothesis now states
		\begin{equation}
		\begin{array}{l}
			\forall \mathit{enclIts}. 
			\bigg(
				\mathit{Reach}(\mathit{start}_{p_1}) 
					\rightarrow \\ \quad
				\Big(
					(
					\llbracket \mathit{Cond} \rrbracket
						\land 
						\trans{F_1}
						)(\mathit{start}_{p_1})
					\rightarrow 
						\trans{F_1}(\mathit{end}_{p_1})
				\Big)
			\bigg). 
			\end{array}\hspace*{-2em}
		\end{equation}
		For the grounding $\sigma'$ we have already established the three premises of this formula, therefore we conclude $\sigma'\trans{F_1}(\mathit{end}_{p_1})$. Since $\mathit{end}_{p_1} = tp_p(\suc(it^p))$, we get $\sigma'\trans{F_1}(tp_p(\suc(it^p))$, which concludes the inductive case.

		We now have established the base case and the
                inductive case, so we use bounded induction to
                conclude $\sigma
                \trans{F_1}(tp_p(\mathit{lastIt}_p))$. Finally, we
                rewrite this fact to $\sigma
                \trans{F_1}(\mathit{end_p})$ using $\sigma'\llbracket
                p \rrbracket$, which shows
                that~\eqref{proof:completeness-formula1} holds.
                \qed
	\end{itemize}
\end{proof}

\section{Correctness of trace lemmas}
\label{sec:trace-lemmas}
We already proved soundness of trace lemma (A1) in
Section~\ref{sec:verification}.
In this section, we prove the remaining two trace lemmas (B1-B2).

\begin{proof}[Soundness of Intermediate Value Trace Lemma (B1)]
	We prove the following equivalent formula obtained from the
        intermediate value trace lemma (B1) by modus tollens.
	\begin{equation}
		\label{form:intermediate-modus-tollens}
	\begin{array}{l}
	\hspace*{-1.5em}\forall x_\Int. \bigg(
			\Big(
				\mathit{Dense}_{w,v} \land
				v(tp_\pv{w}(\zero)) \leq x \;\land 
          \\
         \hspace*{-.5em}\forall it_\Nat.\big(
					(
					it < \mathit{lastIt}_\pv{w}
					\land
						v(tp_\pv{w}(\suc(it))) \eql v(tp_\pv{w}(it)) + 1
					) 
          \\
          \hspace*{-.5em}	\limp
					v(tp_\pv{w}(it)) \neql x
					\big)
			\Big)\\
			\limp 
				v(tp_\pv{w}(\mathit{lastIt}_\pv{w})) \leq x
		\bigg)
	\end{array}\hspace*{-2em}
	\end{equation}
	The proof proceeds by deriving the conclusion of formula \eqref{form:intermediate-modus-tollens} from the premises of formula 
	\eqref{form:intermediate-modus-tollens}.
	
	Consider the instance of the induction axiom scheme with
	\begin{subequations}	
	\begin{align}
		\text{{\scriptsize Base case: }}&v(tp_\pv{w}(\zero)) \leq x 
			\label{form:intermediate-axiom-a}\\
		\text{{\scriptsize Inductive case: }}&\forall it_\Nat. 
			\Big(\big(\zero \leq it < \mathit{lastIt}_\pv{w} \land v(tp_\pv{w}(it)) \leq x \big) \label{form:intermediate-axiom-b} 
				\\ &\limp v(tp_\pv{w}(\suc(it))) \leq x \Big)  \nonumber \\
		\hspace*{-1em}\text{{\scriptsize Conclusion: }}&\forall it_\Nat. \Big(\zero \leq it \leq \mathit{lastIt}_\pv{w} \limp v(tp_\pv{w}(it)) \leq x \Big),\hspace*{-2em}	\label{form:intermediate-axiom-c}
	\end{align}
	\end{subequations} 
	obtained from the bounded induction axiom
        scheme~\eqref{eq:tr:ind} with $P(it) := v(tp_\pv{w}(it)) \leq x$. 
	
	The base case~\eqref{form:intermediate-axiom-a} holds, since it occurs as second premise of formula \eqref{form:intermediate-modus-tollens}. For the inductive case \eqref{form:intermediate-axiom-b}, assume $\zero \leq it < \mathit{lastIt}_\pv{w}$ and $v(tp_\pv{w}(it)) \leq x$. By density of $v$, we obtain two cases: 

	\begin{itemize}
		\item Assume $v(tp_\pv{w}(\suc(it)))=v(tp_\pv{w}(it))$. Since we also assume $v(tp_\pv{w}(it)) \leq x$, we immediately get $v(tp_\pv{w}(\suc(it))) \leq x$.
		\item Assume $v(tp_\pv{w}(\suc(it)))=v(tp_\pv{w}(it)) + 1$. From the assumption $it < \mathit{lastIt}_\pv{w}$ and the third premise of formula \ref{form:intermediate-modus-tollens}, we get $v(tp_\pv{w}(it)) \neql x$, which combined with $v(tp_\pv{w}(it)) \leq x$ and the totality-axiom of $<$ for integers gives $v(tp_\pv{w}(it)) < x$. Finally we combine this fact with $v(tp_\pv{w}(\suc(it)))=v(tp_\pv{w}(it)) + 1$ and the integer-theory-lemma $x<y \limp x+1\leq y$ to derive $v(tp_\pv{w}(\suc(it))) \leq x$.
	\end{itemize}
	Hence, we conclude that the inductive case \eqref{form:intermediate-axiom-b} holds.
	Thus, the conclusion \eqref{form:intermediate-axiom-c} also
        holds. Since the theory axiom $\forall it_\Nat. \ \zero \leq
        it$ holds, formula \eqref{form:intermediate-axiom-c} implies
        the conclusion of formula
        \eqref{form:intermediate-modus-tollens}, which concludes the
        proof.
\qed
\end{proof}

\begin{proof}[Soundness of Iteration Injectivity Trace Lemma (B2)]
	For arbitrary but fixed iterations $it^1$ and $it^2$, assume that the premises of the lemma hold.
	Now consider the instance of the induction axiom scheme with
	\begin{subequations}	
	\begin{align}
		\text{{\scriptsize Base case: }}&v(tp_\pv{w}(it^1)) < v(tp_\pv{w}(\suc(it^1))) \label{form:injectivity-axiom-a}\\
	\text{{\scriptsize Inductive case: }}&\forall it_\Nat. \Big(\big(\suc(it^1) \leq it < \mathit{lastIt}_\pv{w}  \nonumber\\
&		\qquad\quad\land v(tp_\pv{w}(it^1)) < v(tp_\pv{w}(it)) \big) \label{form:injectivity-axiom-b}\\
		&\qquad\limp v(tp_\pv{w}(it^1)) < v(tp_\pv{w}(\suc(it))) \Big) \nonumber\\
	\hspace*{-1em}\text{{\scriptsize Conclusion: }}&\forall
                                                         it_\Nat. \Big(\suc(it^1)
                                                         \leq it \leq
                                                         \mathit{lastIt}_\pv{w}
                                                         \limp
                                                         \nonumber\\
          & \qquad\quad v(tp_\pv{w}(it^1)) < v(tp_\pv{w}(it)) \Big),\label{form:injectivity-axiom-c}
	\end{align}
	\end{subequations} 
	obtained from the bounded induction axiom
        scheme~\eqref{eq:tr:ind} with $P(it) := v(tp_\pv{w}(it^1)) <
        v(tp_\pv{w}(it))$, by instantiating $bl$ and $br$ to
        $\suc(it^1)$, respectively $\mathit{lastIt}_\pv{w}$.

	The base case \eqref{form:injectivity-axiom-a} holds since by integer theory we have $\forall x_\Int.\ x<x+1$ and by assumption $v(tp_\pv{w}(\suc(it^1))) = v(tp_\pv{w}(it^1)) + 1$ holds.

	For the inductive case, we assume for arbitrary but fixed $it$ that $v(tp_\pv{w}(it^1)) < v(tp_\pv{w}(it))$ holds. Combined with $\mathit{Dense}_{w,v}$ and $\forall x_\Int. (x<y \limp x<y+1)$ this yields $v(tp_\pv{w}(it^1)) < v(tp_\pv{w}(\suc(it)))$, so~\eqref{form:injectivity-axiom-b} holds.
	Since both premises~\eqref{form:injectivity-axiom-a}
        and~\eqref{form:injectivity-axiom-b} hold, also the
        conclusion~\eqref{form:injectivity-axiom-c} holds. Next,
        $it^1<it^2$ implies $\suc(it^1)\leq it^2$ (using the
        monotonicity of $\suc$). We therefore have $\suc(it^1)\leq
        it^2 < \mathit{lastIt}_\pv{w}$, so we are able to instantiate
        the conclusion\eqref{form:injectivity-axiom-c} to obtain
        $v(tp_\pv{w}(it^1)) < v(tp_\pv{w}(it^2))$. Finally, we use the
        arithmetic property $\forall x_\Int, y_\Int. (x<y \limp x \neql
        y)$ to conclude $v(tp_\pv{w}(it^1)) \neql v(tp_\pv{w}(it^2))$.
        \qed
\end{proof}

\end{document}